\documentclass{amsart}

\usepackage{mathrsfs}
\usepackage{amssymb}
\usepackage{color}
\usepackage{hyperref}

\newtheorem{theorem}{Theorem}[section]
\newtheorem{proposition}[theorem]{Proposition}
\newtheorem{lemma}[theorem]{Lemma}

\newtheorem{definition}[theorem]{Definition}
\newtheorem{example}[theorem]{Example}
\newtheorem{corollary}[theorem]{Corollary}

\newcommand{\inn}{\mathrm{int}}

\newcommand{\mm}{\mathfrak{Q}_{\mathbb Z}} % Set of NIA martingale measures.
\newcommand{\mmt}{\mathfrak{Q}_{t}} % Set of NIA martingale measures.
\newcommand{\mmm}{\mathfrak{Q}_{\mathbb Z}^{\max}} % Set of NIA martingale measures.
\newcommand{\mmhat}{\hat{\mathfrak{Q}}_{\mathbb Z}} % Set of NIA martingale measures with hat.
\renewcommand{\phi}{\varphi}

\numberwithin{equation}{section}

\begin{document}

\title{Dynamic trading under integer constraints}

\author{Stefan Gerhold, Paul Kr\"uhner}

\date{\today}


%\thanks{This work was financially supported by ...; }

\address{TU Wien
}

\email{sgerhold@fam.tuwien.ac.at, paulkrue@fam.tuwien.ac.at}

\begin{abstract}
  In this paper we investigate discrete time trading under integer
constraints, that is, we assume that the offered goods or shares are traded in
integer quantities instead of the usual real quantity assumption.
For finite probability spaces and rational asset prices this has little
effect on the core of the theory of no-arbitrage pricing. For price processes
not restricted to the rational numbers,
a novel theory of integer arbitrage free pricing and hedging emerges.
We establish an FTAP, involving a set of absolutely continuous martingale
measures satisfying an additional property. The set of prices of a contingent claim
is no longer an interval, but is either empty or dense in an interval.
We also discuss superhedging with integral portfolios.
\end{abstract}

\maketitle

%{\em Keywords:}

%{\em 2010 Mathematics Subject Classifications:} Primary: ; Secondary: .
%\bigskip

\section{Introduction}
% We do all the great stuff.

Classical, frictionless no-arbitrage theory~\cite{DaMoWi90,HaKr79} makes several simplifying assumptions on financial markets. In particular, position sizes may be arbitrary real numbers, which allows trading strategies that cannot be implemented in practice. Even if brokers are receptive to fractional amounts of shares, there will be a smallest fraction that can be purchased or sold. Moreover, traders might wish to avoid odd lots because of additional brokerage fees and the usually poor liquidity of small positions. In this case, the smallest traded unit would be a round lot
consisting of several (e.g., 100) shares. Both situations can be covered by assuming that integer amounts of a price process $(S^1_t,\dots,S^d_t)_{t\in\mathbb T}$ can be traded.
The set of trading times~$\mathbb T$ is assumed to be finite in this paper.
For simplicity, we will call $S^i$ the price process of the $i$-th (risky) asset, although it may have the interpretation of a fraction or a round lot of an actual asset price. We assume that the amount of money in the risk-less asset may take arbitrary real values. On
the one hand, this increases tractability; on the other hand, it makes economic sense, as
the smallest possible modification of the bank account is usually
several orders of magnitude smaller than that of the the risky positions.
Thus, our integer trading strategies in a model with~$d$ risky assets  take values in
$\mathbb R \times \mathbb{Z}^d$ at each time. For some results and proofs,
we also consider rational strategies with values in $\mathbb R \times \mathbb{Q}^d$. By clearing denominators, the corresponding notions of freeness of
arbitrage are equivalent (see Lemma~\ref{le:easy}).
 
To the best of our knowledge, the existing literature on arbitrage, pricing and hedging under trading
constraints \cite{CaPhTo01,FoKr97,FoSc16} invariably imposes convexity assumptions on the set of admissible
strategies, which are unrelated to integrality constraints. The latter do feature prominently
in the computational finance literature, e.g.\ in the papers \cite{BaTr13,Bi96,BoLe09},
which employ mixed-integer nonlinear programming to solve the Markowitz portfolio selection problem.
In the literature, other keywords such as
\emph{minimum lot restrictions}, \emph{minimum transaction level}, and \emph{integral transaction units}
are used with the same meaning as our \emph{integer constraints}.
Somewhat surprisingly, this kind of restriction seems to have received almost no attention
from the viewpoint of no-arbitrage theory. One exception is a paper by Deng et al.~\cite{DeLiWa02},
who show that deciding the existence of arbitrage in a one-period model under integer constraints
is an NP-hard problem.

In our main results, we assume that the underlying probability space is finite
(Assumption {\bf (F)} of Section~\ref{se:setup}). This assumption
is realistic, because actual asset prices move by ticks, and prices
larger than $10^{{10}^{10}}$, say, will never occur. Still, extending our work to arbitrary probability spaces
might be mathematically interesting, but is left for future work.

In Section~\ref{se:setup}, we introduce the notions of no integer arbitrage {\bf (NIA)}
and no integer free lunch {\bf (NIFL)} in a straightforward way.
It turns out (Theorem~\ref{t:NA NIFL}) that the latter property is equivalent to the classical no-arbitrage condition {\bf NA}, and so we concentrate on {\bf NIA} in the rest of the paper.
Our first main result is a fundamental theorem of asset pricing (FTAP; Theorem~\ref{t:real FTAP})
 characterising {\bf NIA}.
It involves a set of absolutely continuous martingale measures satisfying an additional property.
The latter amounts to explicitly avoiding integer arbitrage outside the support of the
absolutely continuous martingale measure. The theorem is thus not as neat as the classical
FTAP, but is still useful for establishing several of our subsequent results. In Section~\ref{se:claims},
we define the set $\Pi_{\mathbb Z}(C)$ of {\bf NIA}-compatible prices of a claim~$C$.
The integer variant of the classical representation using the set of equivalent martingale measures
features only an inclusion instead of an equality (Proposition~\ref{prop:repr Pi}), and in fact $\Pi_{\mathbb Z}(C)$ may be empty.
Even if it is non-empty, it need not be an interval; however, $\Pi_{\mathbb Z}(C)$ is
then always dense in an (explicit) interval, which is the main result of Section~\ref{se:struct}.
As regards methodology, many of our arguments just use the countability
of $\mathbb{Z}^d$ (and $\mathbb{Q}^d$),
or the density of $\mathbb{Q}^d$ in $\mathbb{R}^d$. Still, at some places
(such as Lemma~\ref{l:approximation}, Example~\ref{ex:no cheapest}, and
Theorem~\ref{thm:many claims})
we invoke non-trivial results from number theory, collected in Appendix~\ref{app:nt}.

Readers who are mainly interested in the practical consequences of integer restrictions are invited
to read (besides the basic definitions) Theorem~\ref{thm:rat model},
Theorem~\ref{thm:int inclusion}, and Section~\ref{a:integer superhedging}.
In a nutshell, for the discrete-time models used in practice (finite probability space, floating-point -- i.e., rational -- asset values), the core of no-arbitrage theory does not change much. One exception is the fact that the supremum of claim prices consistent with no-integer-arbitrage need not agree with the smallest integer superhedging price (see Section~\ref{a:integer superhedging}). Still, this property holds in a limiting sense when superhedging a large portfolio of identical claims. That said, our work is by no means the last word on the practical consequences of integer restrictions in dynamic trading. Problems such as quantile hedging, hedging with risk measures,
or hedging under convex constraints may well be worth studying under integer restrictions.
In Section~\ref{se:var opt} we discuss a toy example of variance optimal
hedging under integer constraints, which leads to the closest vector 
problem  (CVP),  a well-known algorithmic lattice problem.

\section{Trading strategies and absence of arbitrage}\label{se:setup}

We will work with a probability space $(\Omega,\mathcal A,P)$. Our main results use the following assumption:
\begin{itemize}
  \item[{\bf (F)}] $\Omega$ is finite, $\mathcal A$ is the power set of $\Omega$, $P[\{\omega\}]>0$ for any $\omega\in\Omega$, and we choose an enumeration $\omega_1,\dots,\omega_n$ of its elements.
\end{itemize}
We assume that there is a finite set of times $\mathbb T := \{0,\dots, T\}$, with $T\in\mathbb N,$ at
which trading may occur, and fix a filtration $(\mathcal F_t)_{t\in\mathbb T}$ where $\mathcal F_T\subseteq \mathcal A$ and $\mathcal F_0=\{\emptyset,\Omega\}$. The (deterministic) riskless interest rate is $r>-1$, and we have~$d$ risky assets with prices
 $S_t = (S_t^1,\dots,S_t^d)$ at time $t\in \mathbb T$, where $S_t$ is assumed to be
 non-negative and $\mathcal F_t$-measurable. The price of the riskless asset is denoted by $S_t^0 := (1+r)^t$ for $t\in\mathbb T$, and we denote the market price processes by $\bar S := (S^0,S)$.

We are interested in trading strategies that consist of integer positions in the risky assets.
All trading strategies we consider are self-financing.

\begin{definition}
    \begin{itemize}
  \item[(i)]
     An \emph{integer (trading) strategy} is a predictable process $(\bar\phi_t)_{t\in\mathbb T\setminus\{0\}}$ 
     with values in $\mathbb R\times \mathbb Z^d$ and $\bar \phi_t \bar S_t=\bar \phi_{t+1} \bar S_t$ for $t\in\mathbb T\setminus\{0,T\}$. For convenience, we will sometimes use the notation $\bar\phi_0:=\bar\phi_1$.
     The set of integer trading strategies is denoted by $\mathcal Z$. 
  \item[(ii)]
    Analogously, we define the set~$\mathcal R$ of all (real) trading strategies
    and the set~$\mathcal Q$ of rational strategies, with values in $\mathbb R\times\mathbb Q^d$.
  \end{itemize}
\end{definition}

We obviously have $\mathcal Z\subseteq\mathcal Q\subseteq\mathcal R$.
For any trading strategy $\bar\phi\in\mathcal R$ we denote its value at time $t\in\mathbb T$ by
\[
  V_t(\bar\phi) := \bar\phi_t \bar S_t=  \sum_{j=0}^d \phi_t^jS_t^j,
\]
and its discounted value by $\hat V_t(\bar\phi):=V_t(\bar\phi)/S_t^0$.
Often it is convenient to work with discounted asset values or discounted gains which are denoted by
 \begin{align*}
    \hat S_t &:= (S_t^1,\dots,S_t^d)/S_t^0, \\
    \Delta \hat S_t &:= \hat S_t-\hat S_{t-1}
 \end{align*}
for $t\in\mathbb T$ resp.\ $t\in\mathbb T\setminus\{0\}$.
The discounted value process then equals
\begin{equation}\label{eq:disc}
  \hat{V}_t(\bar \phi) = V_0(\bar \phi) +\sum_{k=1}^t \phi_k \Delta
  \hat{S}_k,\quad t\in\mathbb T.
\end{equation}

\begin{definition}
    \begin{itemize}
  \item[(i)]
    An \emph{integer arbitrage} is a strategy $\bar\phi \in \mathcal Z$ which is an arbitrage for the market $\bar S$.
  \item[(ii)]
    A model satisfies the \emph{no-integer-arbitrage condition {\bf (NIA)}}, if it admits no integer arbitrage.
    \item[(iii)] Define the set (a $\mathbb{Z}$-module)
    \[
      \mathcal K_{\mathbb Z} := \Big\{ \sum_{k=1}^T \phi_k \Delta \hat{S}_k
      : \bar \phi \in\mathcal Z\Big\}
    \]
    of discounted net gains realizable by integer strategies,
    and \[
      \mathcal{C}_{\mathbb Z}:= (\mathcal K_{\mathbb Z}-L_+^0)\cap L^\infty.
\]
    Assuming {\bf (F)}, we define the condition {\bf NIFL} \emph{(no integer free lunch)} as
    \[
       \mathrm{cl}(\mathcal{C}_{\mathbb Z}) \cap L_+^0  =\{0\}.
    \]
    The closure is taken w.r.t.\ the Euclidean topology, upon identifying~$L^\infty$
    with~$\mathbb{R}^n$.
  \end{itemize}
\end{definition}

Clearly, {\bf NIA} is weaker than the classical no-arbitrage property {\bf NA} or {\bf NIFL}. It turns out that the classical no-arbitrage property {\bf NA} and {\bf NIFL}
are equivalent (for finite probability spaces), see Theorem \ref{t:NA NIFL} below.
The following simple properties will be used often:

\begin{lemma}\label{le:easy}
  \begin{itemize}
   %\margincomment{(i) holds also for multi-period models for finite $\Omega$, because we can find    $m\in\mathbb{N}$ s.t.\ $m\xi_t(\omega)\in\mathbb{Z}^d$ for all $t,\omega$. It also holds for general $\Omega$, see Prop.~\ref{prop:Z Q}.}
  \item[(i)]
    If {\bf (F)} holds, then in the definition of integer arbitrage the condition $\bar\phi\in\mathcal Z$ can be replaced by $\bar\phi\in\mathcal Q$.
  \item[(ii)]
    In the definition of integer arbitrage  the condition
    $V_0(\bar\phi) \leq 0$ can be replaced by
    $V_0(\bar\phi) = 0.$
  \item[(iii)] {\bf NIA} is equivalent to
  \begin{equation}\label{eq:VV}
    \hat V_T(\bar\phi) - \hat V_0(\bar\phi) \geq 0\ \Rightarrow\ \hat V_T(\bar\phi) = \hat V_0(\bar\phi)
  \end{equation}
   for any $\bar\phi\in\mathcal Z$ (or, under {\bf (F)}, for any $\bar\phi\in\mathcal Q$).
%  \item[(iv)] {\bf NIFLVR} is equivalent to
%  \[
%    \forall k\in\mathbb N: \xi_k Y \geq -c_k \quad \Longrightarrow \quad  \lim_{k\rightarrow\infty} P(\xi_kY>0) = 0
%  \]
%   for any sequences $\xi_k\in\mathbb Z^d$, $c_k\in(0,\infty)$, $k\geq 0$ with $c_k\rightarrow 0$.
  \end{itemize}
\end{lemma}
\begin{proof}
  (ii) and~(iii) are proved precisely as in the classical case.
  Part (i): Clearly, any arbitrage strategy in~$\mathcal Z$ is also in~$\mathcal{Q}$. Now assume that there is an arbitrage $\bar\phi$ such that $(\phi_t^1,\dots,\phi_t^d)\in\mathbb Q^d$ for any $t\in\mathbb T$. Define 
   $$ N := \inf\{n \in\mathbb N: n\phi \in \mathbb Z^{d\cdot T}\}. $$
  Then $N\phi_t \in \mathbb Z^d$ for any $t\in\mathbb T$, and $N\bar\phi$ is an arbitrage.
\end{proof}

By~\eqref{eq:disc}, the implication~\eqref{eq:VV} can be written as
\[
  \sum_{k=1}^T \phi_k \Delta\hat{S}_k \geq 0 \ \Rightarrow\ \sum_{k=1}^T \phi_k \Delta\hat{S}_k = 0.
\]
Although our main results assume~{\bf (F)}, we mention that {\bf (F)}
is actually not necessary in parts~(i) and~(iii) of Lemma~\ref{le:easy}.
This follows easily from the fact that arbitrage in a multi-period model
implies the existence of a period that allows arbitrage.
In the classical setup, this is Proposition~5.11 in~\cite{FoSc16}; the proof
works for integer and rational strategies, too.

Under the finiteness condition {\bf (F)} on $\Omega$,
we can show that any real trading strategy can be approximated by an integer
trading strategy  with a certain rate. The proof is based on Dirichlet's approximation theorem
(Theorem~\ref{thm:dir}).
\begin{lemma}\label{l:approximation}
  \begin{itemize}
  \item [(i)]
  If $S$ is bounded, then for any strategy $\bar\phi\in\mathcal R$ and any $\epsilon>0$,
  we can find a  strategy $\bar\psi\in\mathcal Q$ such that
 $$     \sup_{t\in\mathbb T}\, \mathrm{ess\, sup}\,|V_t(\bar\phi)-V_t(\bar\psi)| < \epsilon $$
and $V_0(\bar\phi)=V_0(\bar\psi)$.
   \item[(ii)]
  Assume {\bf (F)} and let $\bar\phi\in\mathcal R$ and $\epsilon>0$. Then there is $q\in\mathbb N$ and a strategy $\bar\psi\in\mathcal Z$ such that $V_0(\bar\phi)=V_0(\bar \psi)/q$ and
    $$ \sup_{\substack{t\in\mathbb T,j=1,\dots d,\\l=1,\dots,n}} |\psi_t^j(\omega_l) - q\phi_t^j(\omega_l)| < q^{-1/(nd(T+1))} <\epsilon. $$
  In particular, for any strategy $\bar\phi\in\mathcal R$ we can find strategies $\bar\psi\in\mathcal{Q}$, $\bar\eta\in\mathcal Z$ and $q\in \mathbb N$ such that
  \begin{align*}
      \sup_{t\in\mathbb T}\, \mathrm{ess\, sup}\,|V_t(\bar\phi)-V_t(\bar\psi)| &< \epsilon, \\
      \sup_{t\in\mathbb T}\, \mathrm{ess\, sup}\,|qV_t(\bar\phi)-V_t(\bar\eta)| &< \epsilon
  \end{align*}
   and $V_0(\bar\phi) = V_0(\bar\psi) = V_0(\bar\eta)/q$.
   \end{itemize}
\end{lemma}
\begin{proof}
  The first part is trivial as any real number can be approximated by rational numbers. Thus we find a sequence of strategies $(\bar\psi^{(k)})_{k\in\mathbb N}$ in $\mathcal Q$ such that $\psi^{(k)} \rightarrow \phi$ uniformly in $\omega$ for $k\rightarrow\infty$. This and the boundedness of $S$ imply the convergence of the value at any time $t\in\mathbb T$ if the initial value is being fixed as equal.

  To show part~(ii), let $R_t := \{ \phi_t^j(\omega_l): j=1,\dots, d,l=1,\dots,n\}$. For any $t\in\mathbb T$ let $a_t^1,\dots,a_t^{K_t}$ be an enumeration of the elements of $R_t$. We have $K_t\leq dn$ for any $t\in\mathbb T$ and thus $\sum_{t\in\mathbb T}K_t \leq nd(1+T)$. By Dirichlet's approximation
  theorem (Theorem~\ref{thm:dir}), we find $q\in\mathbb N$ with $q^{-1/(nd(1+T))} < \epsilon$ and $p_t^k \in \mathbb Z$ with $|p_t^k-qa_t^k| < q^{-1/(nd(1+T))}$ for any $t\in\mathbb T$, $k=1,\dots,K_t$. For $t\in \mathbb T$ we define
   $$ \psi_t^j(\omega_l) := p_t^k $$
 where $k\in\{1,\dots,K_t\}$ is such that $\phi_t^j(\omega_l) = a_t^k$. Then $$\{\psi_t^j = p_t^k\} \subseteq \bigcup_{m \in A_k} \{\phi_t^j = a_t^m\} $$
  where $A_k = \{m=1,\dots, K_t: p_t^m = p_t^k\}$. Thus $\psi_t$ is measurable w.r.t.\ to the
  $\sigma$-algebra generated by $\phi_t$
   and, hence, $\mathcal F_{t-1}$-measurable. Therefore, $\psi$ is a predictable $\mathbb Z^d$-valued process. The uniform distance of $\psi$ and $\phi$ is less than $q^{-1/(nd(1+T))}$ by construction.
\end{proof}

With the previous lemma at hand we can show that under the finiteness condition classical no-arbitrage is equivalent to {\bf NIFL}.
\begin{theorem}\label{t:NA NIFL} %\margincomment{(ii) implies (iii) also for unrestricted $\Omega$. In (iii) implies (ii) we use the finiteness of $\Omega$ to imply that $Y$ is bounded. With bounded $Y$ the reverse implication works out too. The multi-period case works with the same proof.}
 Assume {\bf (F)}. Then the following statements are equivalent:
  \begin{itemize}
     \item[(i)] There is an equivalent martingale measure $Q\approx P$,
     \item[(ii)] The model satisfies  the classical no-arbitrage property {\bf NA} and
     \item[(iii)] The model satisfies {\bf NIFL}.
  \end{itemize}
 Moreover, if the number of risky assets is $d=1$, then the following statement is equivalent as well:
  \begin{itemize}
    \item[(iv)] The model satisfies {\bf NIA}.
  \end{itemize}
\end{theorem}
\begin{proof}
  The equivalence of (i) and (ii) is the classical FTAP, see \cite[Theorem 5.16]{FoSc16}. Furthermore,  {\bf NA} is equivalent to the classical no free lunch
  condition in our setup (see~\cite{DaMoWi90,KaSt01}), which yields the implication (ii)$\Rightarrow$(iii).
  
Now we assume (iii) and show (ii). Let $\bar\phi\in\mathcal R$ such that $V_0(\bar\phi) = 0$ and $V_T(\bar\phi) \geq 0$.
%We may assume that $C := \sup\{S^j_T(\omega):\omega\in\Omega,j=1,\dots,d\} > 0$ because otherwise $V_T(\bar\phi) = V_0(\bar \phi)=0$ as desired.
By part~(ii)
of Lemma~\ref{l:approximation} we find $q_N\in\mathbb N$ and strategies
$\psi^{(N)}\in\mathcal Z$ such that
\begin{equation}\label{eq:appl dir}
 \mathrm{ess\, sup}\, |q_N V_T(\bar\phi)-V_T(\bar\psi^{(N)})| \leq\frac1N.
\end{equation}
W.l.o.g., the sequence~$q_N$ increases.
We get
 \begin{equation}\label{eq:1N}
    V_T(\bar\psi^{(N)}) \geq q_NV_T(\bar\phi) - \frac1N \geq -\frac1N.
 \end{equation}
 Define
 \begin{align*}
   Z_N &:=
   \begin{cases}
     1 & V_T(\bar\psi^{(N)})  > 1, \\
     V_T(\bar\psi^{(N)}) & V_T(\bar\psi^{(N)}) \leq 1
   \end{cases} \\
   &= V_T(\bar\psi^{(N)}) -(V_T(\bar\psi^{(N)}) -1) \,1_{\{V_T(\bar\psi^{(N)}) >1\}} \in \mathcal{C}_{\mathbb Z}, \quad N\in\mathbb N.
 \end{align*}
 Since $Z_N\in L^\infty$, there is a convergent subsequence, and w.l.o.g.
 $Z_N$ itself converges to some $Z\in \mathrm{cl}(\mathcal{C}_{\mathbb Z})$.
 By~\eqref{eq:1N}, we have $Z\geq0$. Then, {\bf NIFL} implies that
 $Z=0$, and thus $V_T(\bar\psi^{(N)})  \to0$. Since $q_N^{-1}V_T(\bar\psi^{(N)}) \to V_T(\bar \phi)$ by~\eqref{eq:appl dir} (recall that~$q_N$ increases), we conclude
 that $V_T(\bar \phi) =0$.

Now assume that $d=1$. (ii)$\Rightarrow$(iv) is obvious. We assume that~(ii) does not hold.
Proposition~5.11 in~\cite{FoSc16} yields the existence of a one-period arbitrage, i.e.\ an arbitrage $\bar\phi$ and $t_0\in\mathbb T$ such that $\phi_t=0$ for any $t\in\mathbb T\setminus\{t_0\}$. Since $\bar\phi$ is an arbitrage we must have $\phi^1_{t_0}\neq 0$. Define
 $$ \psi^j_t := \phi^t_j / |\phi^1_{t_0}|,\quad t\in\mathbb T,j=0,1. $$
Then $\psi$ is an arbitrage as well. Moreover, $\psi_t^1\in\{-1,0,1\}\subseteq\mathbb Z$ for any $t\in\mathbb T$, thus $\psi\in\mathcal Z$. Consequently, (iv) does not hold.
\end{proof}
In practice, all values occurring in the model specification are floating-point numbers. The following result
shows that in this case the existence of an arbitrage opportunity is not affected
by integrality constraints. 
\begin{theorem}\label{thm:rat model}
  Assume {\bf (F)}, and that the interest rate~$r$ and all asset values are rational:
  $r\in\mathbb{Q}$, and $S_t\in\mathbb{Q}^d$ for $t\in\mathbb{T}$.
  Then {\bf NIA} is equivalent to {\bf NA}.
\end{theorem}
\begin{proof}
   {\bf NA} always implies {\bf NIA}. Now suppose that we have a real arbitrage opportunity.
   By part~(iii) of Lemma~\ref{le:easy}, there is a predictable process $\phi$ such that
   \begin{align*}
      \forall\, \omega \in\Omega:\ \sum_{k=1}^T \phi_k(\omega) \Delta \hat{S}_k(\omega)&\geq 0, \\
      \exists\, \omega \in\Omega:\ \sum_{k=1}^T \phi_k(\omega) \Delta \hat{S}_k(\omega)&> 0.
   \end{align*}
   The assertion now follows from Lemma~\ref{le:rat sol} below. Note that predictability
   of the resulting rational process is easy to guarantee, by introducing for all $k,j$ a \emph{single}
   variable for the
   $\phi_k^j(\omega)$ for which the $\omega$s belong to the same atom of $\mathcal{F}_{k-1}.$
\end{proof}

In the proof of the preceding result, we applied the following simple lemma.
Using Ehrhart's theory of lattice points in dilated polytopes~\cite{BeRo15,St12}, it is certainly
possible to state much more general results along these lines.\footnote{We thank Manuel Kauers
for pointing this out.}
Therefore, we do not claim originality for Lemma~\ref{le:rat sol}, but give a short
self-contained proof for the reader's convenience.

\begin{lemma}\label{le:rat sol}
  Let $(a_{ij})_{1\leq i\leq r,1\leq j\leq s}$ be a matrix with rational
  entries $a_{ij}\in\mathbb Q$. Suppose that
  there is a real vector $(x_1,\dots,x_s)$ such that
  \begin{equation}\label{eq:sys}
    \sum_{j=1}^s a_{ij} x_j \geq 0, \quad i=1,\dots,r,\quad
    \text{with at least one inequality being strict.}
  \end{equation}
  Then there is a rational vector satisfying~\eqref{eq:sys}.
\end{lemma}
\begin{proof}
  After possibly reordering the lines of the matrix $(a_{ij})$, we may assume that there
  is $u\in\{1,\dots,r\}$ such that
  \begin{align*}
    \sum_{j=1}^s a_{ij} x_j > 0, \quad 1\leq i\leq u, \\
    \sum_{j=1}^s a_{ij} x_j = 0, \quad u<i\leq r.
  \end{align*}
  By defining $y_i:=\sum_{j=1}^s a_{ij} x_j$ for $1\leq i\leq u$, we get that
  the vector $(x_1,\dots,x_s,y_1,\dots,y_u)$ solves the system
  \begin{align}
    \sum_{j=1}^s a_{ij} x_j - y_i &= 0, \quad 1\leq i\leq u, \label{eq:sys1} \\
    \sum_{j=1}^s a_{ij} x_j &= 0, \quad u<i\leq r, \label{eq:sys2} \\
    y_1,\dots,y_u&>0. \label{eq:sys ineq}
  \end{align}
  Equations~\eqref{eq:sys1} and~\eqref{eq:sys2} constitute a homogeneous linear system of equations
  with rational coefficients, which has a basis~$B\subset \mathbb{Q}^{s+u}$ of rational
  solution vectors, by Gaussian elimination. The vector $(x_1,\dots,x_s,y_1,\dots,y_u)$
  can be written as a linear combination of vectors in~$B$.
  By approximating the (real) coefficients of this linear combination with rational
  numbers, we get a vector
  $(\tilde{x}_1,\dots,\tilde{x}_s,\tilde{y}_1,\dots,\tilde{y}_u)\in\mathbb{Q}^{s+u}$
  satisfying \eqref{eq:sys1}--\eqref{eq:sys ineq}. Then
  $(\tilde{x}_1,\dots,\tilde{x}_s)$ is the desired rational vector.
\end{proof}

The assertion of Theorem~\ref{thm:rat model} does not hold for infinite probability spaces,
as the following example illustrates.

\begin{example}
  Let $\Omega=\{\omega_1,\omega_2,\dots\}$ be countable, $\mathcal A=2^\Omega$, and fix an
  arbitrary probability measure~$P$ with $P[\{\omega_i\}]>0$ for all~$i\in\mathbb N$.
  We choose $d=2$, $T=1$, and $r=0$. The asset prices are defined by $S_0=(1,1)$
  and
  \[
    S_1(\omega_i)=
    \begin{cases}
       (1+p_i,1+q_i) & i\ \text{even}, \\
       (1-\hat{p}_i,1-\hat{q}_i) & i\ \text{odd},
    \end{cases}
  \]
  where $p_i,q_i,\hat{p}_i,\hat{q}_i$ are natural numbers satisfying
  \begin{equation}\label{eq:pq}
    \frac{p_i}{q_i} \searrow \pi \quad \text{and} \quad \frac{\hat{p}_i}{\hat{q}_i} \nearrow \pi,\quad
    i\to\infty.
  \end{equation}
  Thus, the increments are
  \[
    \Delta S_1(\omega_i)=
    \begin{cases}
       (p_i,q_i) & i\ \text{even}, \\
       (-\hat{p}_i,-\hat{q}_i) & i\ \text{odd}.
    \end{cases}
  \]
  A vector $(\phi_1^1,\phi_1^2)\in\mathbb{R}^2$ yields an arbitrage
  if and only if
  \begin{align}
    \frac{p_i}{q_i}\phi_1^1+\phi_1^2 &\geq 0,\quad i\ \text{even}, \label{eq:even} \\
    \frac{\hat{p}_i}{\hat{q}_i}\phi_1^1+\phi_1^2 &\leq 0,\quad i\ \text{odd,} \label{eq:odd}
  \end{align}
  with at least one inequality being strict.
  By~\eqref{eq:pq}, the vector $(\phi_1^1,\phi_1^2)=(1,-\pi)$ satisfies this,
  and so {\bf NA} does not hold. By letting~$i$ tend to infinity, we see that
  there is no integer vector satisfying \eqref{eq:even}-\eqref{eq:odd},
  which shows that the model satisfies {\bf NIA}.
\end{example}

Our next goal is to characterise {\bf NIA}, without restricting the asset prices to rational
numbers. As we will see, for $d>1$ {\bf NIA} is not equivalent to the existence of an equivalent martingale measure, but rather to the existence of an
\emph{absolutely continuous} martingale measure with an additional property.
We first introduce sets of strategies which do not yield any net profit.

\begin{definition}\label{def:Q}
  \begin{itemize}
  \item[(i)]
  Let $Q$ be a probability measure on $(\Omega,\mathcal A)$, and denote the set of trading strategies with zero initial value by $\mathcal R_0:=\{\bar\phi\in\mathcal R:V_0(\bar\phi)=0\}$. We denote the set of all integer-valued (resp.\ rational-valued, resp.\ real-valued) trading strategies with zero initial capital and $Q$-a.s.\ zero gain by
    \begin{align*}
     \mathcal Z^0_Q &:= \{ \bar\phi\in \mathcal R_0 \cap \mathcal Z : V_T(\bar\phi)=0\ \ Q\text{-a.s.}\}, \\
     \mathcal Q^0_Q &:= \{ \bar\phi\in \mathcal R_0 \cap \mathcal Q: V_T(\bar\phi)=0\ \ Q\text{-a.s.}\}, \\
     \mathcal R^0_Q &:= \{ \bar\phi\in \mathcal R_0: V_T(\bar\phi)=0\ \ Q\text{-a.s.}\}.
    \end{align*}
%   and denote the orthogonal projection from the set of all predictable process to $Z_Q$ relative to the norm. \margincomment{???}
  \item[(ii)] If we assume\footnote{This ensures that the sets $\{\omega\}$ occurring
  in part~(ii) of Definition~\ref{def:Q} are measurable.} {\bf (F)} then we write $\mm^{\max}$ for the set of martingale measures $Q\ll P$ such that
        \begin{align*}
           &\forall \bar\phi \in \mathcal Q^0_Q:\left(  V_T(\bar\phi)\geq 0 \Rightarrow  V_T(\bar\phi) = 0\right)\quad\text{and} \\
           &\exists \bar\phi \in \mathcal R^0_Q: V_T(\bar\phi)\geq 0\text{ and } \{V_T(\bar\phi) > 0\} = \{\omega\in\Omega: Q[\{\omega\}] = 0 \}.
        \end{align*}
  \item[(iii)] $\mm$ denotes the set of martingale measures $Q\ll P$ such that
  \begin{align*}
           \forall \bar\phi \in \mathcal Z^0_Q:\left(  V_T(\bar\phi)\geq 0 \Rightarrow  V_T(\bar\phi) = 0\right).
        \end{align*}
  \end{itemize}
\end{definition}
Obviously, we have $\mathfrak{Q}\subseteq\mm^{\max} \subseteq \mm$, where $\mathfrak{Q}$
denotes the set of equivalent martingale measures. (As for the first inclusion, $\bar \phi=0$
satisfies the existence statement in~(ii).) Before giving an FTAP for integer trading we show further properties of the measures in $\mm^{\max}$.
\begin{proposition}\label{p:Qmax}
 Assume {\bf (F)} and that $\mm^{\max} \neq \emptyset$. Then there is a set $A\subsetneq\Omega$ such that $\mm^{\max}$ is the set of martingale measures whose support is $\Omega\setminus A$. Also, there is $\bar\phi\in\mathcal R_0$ with $V_T(\bar\phi)\geq 0$ and $\{V_T(\bar\phi)>0\} =A$.
 
 Now, let $Q\in \mm^{\mathrm{max}}$ and let $Q'\in \mm$. Then $Q'\ll Q$. Moreover, $\mm^{\mathrm{max}}$ is dense in $\mm$ with respect to the total variation distance. 
\end{proposition}
\begin{proof}
  Choose $Q\in\mm^{\max}$ and let $\bar\phi\in\mathcal R_Q^0$ satisfy the existence statement in (ii) of Definition \ref{def:Q}. Define $A:=\{V_T(\bar\phi)>0\}$. Then~$\bar\phi$ is the required trading strategy.

  Let $Q'$ be a martingale measure with support equal to $\Omega\setminus A$. Then $\bar\phi$ satisfies the existence statement of (ii) in Definition \ref{def:Q}. Let $\bar\psi\in\mathcal Q_{Q'}^0$ with $V_T(\bar\psi)\geq 0$. Then $V_T(\bar\psi) = 0$ $Q'$-a.s., i.e.\ $V_T(\bar\psi) = 0$ on $\Omega\setminus A$. Consequently, $\bar\psi\in\mathcal Q_{Q}^0$. (ii) of Definition \ref{def:Q} yields that $V_T(\bar\psi) = 0$. Thus, $Q'\in\mm^{\mathrm{max}}$.
  
 We need to show that any measure in $\mm^{\max}$ is a martingale measure with support $\Omega\setminus A$. This, however, follows as soon as we have shown that $Q'\ll Q$ for any $Q'\in\mm$. Let $Q'\in\mm$. Observe that $V_T(\bar\phi) = 0$ $Q'$-a.s.\ because $Q'$ is a martingale measure. Thus $A = \{V_T(\bar\phi) > 0\}$ is a $Q'$-null set. We find $Q'\ll Q$.
 
 Finally, we have to show that $Q'$ can be approximated by elements in $\mm^{\max}$ in total variation. Define $Q_\alpha := \alpha Q'+(1-\alpha)Q$ for any $\alpha\in[0,1]$. Then $Q' = Q_1 \leftarrow Q_\alpha$ as $\alpha\rightarrow 1$. However, $Q_\alpha$ is a martingale measure with the same support as $Q$ for $\alpha\neq 1$ and, hence, it is in $\mm^{\max}$ by what we have shown so far.
\end{proof}

We can now state an FTAP
for integer trading.
\begin{theorem}\label{t:real FTAP}
   Assume {\bf (F)}. Then the following statements are equivalent:
    \begin{itemize}
      \item[(i)] $\mm^{\max}\neq \emptyset$
      \item[(ii)] $\mm\neq \emptyset$
      \item[(iii)] The market satisfies {\bf NIA}.
    \end{itemize}
  The implication (ii)$\Rightarrow$(iii) does not need assumption {\bf (F)}.
\end{theorem}
\begin{proof}
  (i)$\Rightarrow$(ii) is trivial.
  
  (ii)$\Rightarrow$(iii): We fix a measure $Q\in\mm$. Let $\bar\phi\in\mathcal R_0\cap \mathcal Z$ with $V_T(\bar\phi)\geq 0$. Since $Q$ is a martingale measure we have $\hat V_T(\bar\phi) = 0$ $Q$-a.s.\ and, hence, $V_T(\bar\phi) = 0$ $Q$-a.s. Thus $\bar\phi\in \mathcal Z^0_Q$. By part~(iii) of
  Definition~\ref{def:Q} we have $V_T(\bar\phi) = 0$. Hence, we have {\bf NIA}.
  
  (iii)$\Rightarrow$(i): Let
  \[
    A:=\{\omega\in\Omega: \exists \bar\phi\in 
      \mathcal R_0: V_T(\bar\phi) \geq 0\wedge V_T(\bar\phi)(\omega) > 0\}.
   \]
    For every $\omega\in A$ choose an according strategy $\bar\phi^{(\omega)}\in \mathcal R_0$ with $V_T(\bar\phi^{(\omega)}) \geq 0$ and $V_T(\bar\phi^{(\omega)})(\omega) > 0$. Define
    \[
      \bar\phi := \sum_{\omega\in A} \bar\phi^{(\omega)}.
    \]
    Then $\bar\phi \in \mathcal R_0$, $V_T(\bar\phi)\geq 0$ and $\{V_T(\bar\phi) > 0\} = A$.
  
  We claim that for any $\bar\psi \in \mathcal R_0$ with $V_T(\bar\psi)1_{\Omega\setminus A}\geq 0$ we have $V_T(\bar\psi)1_{\Omega\setminus A} = 0$. Let $\bar\psi \in \mathcal R_0$ with $V_T(\bar\psi)1_{\Omega\setminus A}\geq 0$. If $V_T(\bar\psi) \geq 0$ on $A$, then $V_T(\bar\psi) \geq 0$ and, hence, $V_T(\bar\psi)1_{\Omega\setminus A} = 0$ by construction of $A$. Thus, we may assume that $V_T(\bar\psi)(\omega) < 0$ for some $\omega\in A$. Then
  \[
    c:= -\frac{\min\{V_T(\bar\psi)(\omega):\omega\in A\}}{\min\{V_T(\bar\phi)(\omega):\omega\in A\}}> 0.
  \]
  The strategy $\bar\psi+c\bar\phi$ is in $\mathcal R_0$ and satisfies $V_T(\bar\psi+c\bar\phi)\geq 0$. Thus, $V_T(\bar\psi+c\bar\phi) = 0$ outside $A$. Hence, $V_T(\bar\psi) = 0$ outside $A$, i.e.\ $V_T(\bar\psi)1_{\Omega\setminus A} = 0$.
  
   Assume for contradiction that $A=\Omega$. Then $V_T(\bar\phi)>0$. Define 
   \[
     e := \min\{ V_T(\bar\phi)(\omega): \omega\in \Omega\} > 0.
   \]
   Lemma \ref{l:approximation} yields $q\in\mathbb N$ and $\bar\psi\in \mathcal Z$ such that $|V_T(\bar\psi) - V_T(q\bar\phi)| < e$. Thus, $V_T(\bar\psi) > qV_T(\bar\phi)-e \geq 0$. Thus $\bar\psi$ is an integer arbitrage. A contradiction.
   
   Consequently, $A\subsetneq \Omega$. We have shown that the market $\bar S$ is free of arbitrage on $\Omega\setminus A$. The classical fundamental theorem yields a martingale measure $Q$ on $\Omega\setminus A$ for $\bar S$. We denote its extension to a probability measure on $\Omega$ by $Q$, i.e. $Q[M] = Q[M\setminus A]$ for any $M\subseteq \Omega$. Then $Q\ll P$ and $Q$ is a martingale measure with $\{\omega\in \Omega: Q[\{\omega\}] = 0\} = A$. Since $\{V_T(\bar\phi)>0\} = A$ we have the existence statement in
   part~(ii) of Definition~\ref{def:Q}.
   Now let $\bar\psi \in \mathcal Q_Q^0$ with $\hat V_T(\bar\psi)\geq 0$. Let $q$ be a common denominator for $\{ \psi_t^j(\omega_l): t\in\mathbb T,\, j=1,\dots,d,\, l=1,\dots,n\}$. Then $q\bar\psi \in \mathcal Z_Q^0$. Since we have {\bf NIA} we get $V_T(\bar\psi) = \frac1qV_T(q\bar\psi) = 0$ as claimed.
\end{proof}

An immediate consequence is the following sufficient criterion for the construction of markets with no integer arbitrage.
\begin{corollary}\label{k:konstr}
  Let $Q\ll P$ be a martingale measure and assume that $\mathcal Z_0^Q = \{0\}$. Then the market satisfies {\bf NIA}.
\end{corollary}

The following example is a simple application of the preceding corollary.
\begin{example}
  Assume that $d=2$, $n\geq 2$, $T=1$, $r=0$ and choose  $(S^1_0,S^2_0)=(1,\pi)$ and
  \[
    (S_1^1,S_1^2)(\omega_j):=
    \begin{cases}
      (3/2,3\pi/2) & j=1, \\
      (1/2,\pi/2) & j=2.
    \end{cases}
  \]
  Define $Q[\{\omega_j\}] = 1_{\{j=1,2\}}/2$ for $j=1,\dots,n$. Then $Q\ll P$ is a martingale measure, and %$\mathcal{R}^0_Q = (\pi,-1)\mathbb R$.
  $\mathcal{R}^0_Q =\{(0,-\pi\phi^2,\phi^2):\phi^2\in\mathbb R\}.$ Consequently, we have $\mathcal{Z}^0_Q=\{0\}$. Thus, Corollary~\ref{k:konstr} yields that the market does not allow for integer arbitrage.  
  Observe that this holds regardless of the specification of $(S_1^1,S_1^2)(\omega_j)$ for $j\geq 3$.
\end{example}

Another immediate consequence is the existence of absolutely continuous martingale measures.
\begin{corollary}%\margincomment{Can be proven with countable $\Omega$ and bounded $S$ as well.}
  Suppose that a model satisfies {\bf NIA} and assume {\bf (F)}. Then there is an absolutely continuous martingale measure.
\end{corollary}
\begin{proof}
  %If the model is arbitrage-free (in the classical sense), then there
  %is an \emph{equivalent} martingale measure by the classical FTAP.
 Immediate from Theorem \ref{t:real FTAP} (iii)$\Rightarrow$(i).
\end{proof}

The following example shows that the existence of an absolutely continuous martingale measure alone is insufficient to exclude integer arbitrage.
\begin{example}
  Let $\Omega = \{\omega_1,\omega_2\}$, $S^0_0 = 1 = S^0_1$, $S^1_0= 1$ and $ S_1^1(\omega_i) = i $ for $i=1,2$ (i.e.\ $T=1$, $d=1$, $n=2$). Then $Q := \delta_{\omega_1}$ is a martingale measure which is absolutely continuous with respect to $P := (\delta_{\omega_1} + \delta_{\omega_2})/2$, where $\delta_{\omega_j}$ denotes the Dirac-measure on $\omega_j$. The strategy $\bar\phi_1 := (-1,1)$ is an integer arbitrage.
\end{example}

Finally, we provide a technical statement that will be used in Section~\ref{se:struct}.
\begin{lemma}\label{l:two nodes}
  Assume {\bf (F)}, let $Q\in\mmm$ and assume that for any $B\in\mathcal F_1$ we have $Q[B]\in\{0,1\}$. Then $\mathcal F_1=\mathcal F_0$.
\end{lemma}
\begin{proof}
  Let $A\in\mathcal F_1$ be maximal with $Q[A]=0$. Then $B:=\Omega\setminus A$ is an atom, and its only strict subset contained in $\mathcal F_1$ is the empty set. If $A=\emptyset$, then the claim follows trivially. Assume for contradiction that $A\not= \emptyset$. We claim that the model restricted to $A$ still satisfies {\bf NIA}.
  To this end let $\bar\phi\in\mathcal Q$, $t=1,\dots T$ with $\bar\phi_1=\dots \bar\phi_{t-1}=0$ and $\hat V_t(\bar\phi)=\dots =\hat V_T(\bar\phi)\geq 0$ on $A$.
  (Since existence of an arbitrage implies existence of a one period arbitrage,
  it suffices to consider this kind of strategy.)
  
  \emph{Case 1:} $t=1$. Since $\hat S_0 = E_Q[\hat S_1] = \hat S_1(B)$, we find that $\hat V_1(\bar\phi)(B) = V_0(\bar\phi) = 0$ and, hence, $V_1(\bar\phi)\geq 0$ everywhere. Since $(S^0,\dots,S^d)$ satisfies {\bf NIA}, we obtain that $V_1(\bar\phi) = 0$ and, hence, $V_s(\bar\phi)=0$ for any $s\in\mathbb T$.
  
  \emph{Case 2:} $t\geq 2$. Define $\bar\psi := 1_A\bar\phi$. Since $A\in\mathcal F_1$ and $\bar\phi_1=0$ we find that $\bar\psi\in \mathcal Q$ with $\bar\psi_0=\dots \bar\psi_{t-1}=0$ and $\hat V_t(\bar\psi)=\dots =\hat V_T(\bar\psi)\geq 0$. Since the model $(S^0,\dots,S^d)$ on $\Omega$ satisfies {\bf NIA} by assumption we find that $0 = V_t(\bar\psi) = 1_A V_t(\bar\phi)$ and, hence, $V_t(\bar\phi) = 0$ on $A$. 
  
  Thus $(S^0,\dots,S^d)$ restricted to $A$ satisfies {\bf NIA}. By Theorem~\ref{t:real FTAP} there is $Q'\in \mmm$ for the model $(S^0,\dots,S^d)$ restricted to $A$. We denote its extension to $\Omega$ by $Q'$
  as well, i.e.\ $Q'[C] = Q'[C\cap A]$ for any $C\in\mathcal A$. Define $\tilde Q:=Q/2+Q'/2$ and observe that $\tilde Q\in \mm$. However, $Q \not\approx \tilde Q$ because $Q'$ has disjoint support with $Q$.
  But Proposition~\ref{p:Qmax} implies $Q\approx \tilde Q$, which yields
  a contradiction. Thus $A=\emptyset$ and, hence, $\mathcal F_1=\mathcal F_0$.
\end{proof}

\section{Claims and integer trading}\label{se:claims}
\begin{definition}
  Fix a model that satisfies {\bf NIA}.
  \begin{itemize}
  \item[(i)]A \emph{claim} is a random variable $C\geq0$. A real number $p\geq0$ is an
  \emph{integer arbitrage free price} of~$C$, if there is an adapted non-negative stochastic process $(X_t)_{t\in\mathbb T}$ with $X_0=p$, $X_T=C$ such that the market 
    $$ (S^0,\dots,S^d,X)$$
  satisfies  {\bf NIA}. The set of integer arbitrage free prices is denoted
  by $\Pi_{\mathbb{Z}}(C)$.
  \item[(ii)]
  An \emph{integer superhedge} for $C$ is a trading strategy $\bar \phi \in\mathcal Z$ such that $V_T(\bar\phi) \geq C$,
  and it is an \emph{integer replication strategy} if it satisfies $V_T(\bar\phi) = C.$
  We write
  \begin{equation}\label{eq:def super}
    \sigma_{\mathbb{Z}}(C) = \inf \{V_0(\bar \phi): \bar\phi\in\mathcal Z,\ V_T(\bar \phi)\geq C\}
  \end{equation}
  for the infimum of prices of integer
superreplication strategies for~$C$.
\end{itemize}
\end{definition}

Analogously to $\Pi_{\mathbb{Z}}(C)$, we write $\Pi(C)$ for the set of classical arbitrage free prices
in models satisfying {\bf NA}.
We recall the classical superhedging theorem (Corollaries~7.15 and 7.18
in~\cite{FoSc16}):
\begin{theorem}\label{thm:super}
  Assume that {\bf NA} holds, and let~$C$ be a claim with $\sup \Pi(C)<\infty$.
  Then there is a strategy $\bar\phi\in\mathcal R$ with $V_0(\bar\phi)=\sup \Pi(C)$
  and $V_T(\bar\phi)\geq C$. Moreover, $\sup \Pi(C)$ is the smallest number with this property.
\end{theorem}

We find analogue statements to the preceding theorem under the weaker assumption {\bf NIA}. Proposition \ref{p: superhedge} below states that {\bf NIA} suffices for the existence of a real cheapest superhedge whose price is the infimum of all rational superhedging prices. Moreover, Theorem \ref{thm:dense} below implies that either the set of {\bf NIA} compatible prices for the claim is empty, or its supremum equals the cheapest superhedging price.

There is no need to define the notion of integer completeness, because there
would be no interesting models that have this property:

\begin{proposition}
  %\margincomment{Multi-period case, any $\Omega$. We do not even need {\bf NIA}.}
  The following statements are equivalent:
  \begin{itemize}
    \item[(i)] Every claim is replicable by an integer strategy,
    \item[(ii)] The probability space $(\Omega,\mathcal A,P)$ consists of a single atom.
  \end{itemize}
\end{proposition}
\begin{proof}
  If~(ii) holds and~$C$ is a claim, then there is a constant $c\in[0,\infty)$
  such that $C=c$ a.s. Then~$C$ is replicated by the integer strategy
  $\bar\phi=(\phi^0,0)$ with $\phi^0_t=c/(1+r)^{T-t}$, $t\in\mathbb T$.
  
  Now suppose that every claim is integer replicable. In particular, then,
  each claim is replicable in the classical sense. It is well known that this
   implies that $\Omega$ has a partition into finitely many atoms.
  (This result is, of course, usually proved in the framework of a model
  satisfying {\bf NA}. Assuming {\bf NA} is not necessary though, as seen from the proof of Theorem~5.37 in~\cite{FoSc16}.) If $\Omega$ does not consist of a
  single atom, then we can fix two distinct atoms $A$ and $B$. For a random variable $X$ we can find its essential value on $A$ (resp.\ $B$) and denote it by $\delta_A(X)$ (resp.\ $\delta_B(X)$). A self-financing integer trading strategy~$\bar\phi$ is uniquely defined by specifying its initial wealth~$V_0(\bar\phi)$ and the predictable $\mathbb{Z}^d$-valued process $\phi=(\phi_t)_{t=1,\dots,T}$. Thus there is a bijective map
   $$ \Gamma: \mathbb R\times \mathcal Z_{\mathrm{c}} \rightarrow \mathcal Z $$
 where
  $\mathcal Z_{\mathrm{c}} := \{ (\phi^1,\dots,\phi^d): \bar\phi\in\mathcal Z\}$ is countable with $V_0(\Gamma(v,\phi)) = v$ for any $v\in\mathbb R$. In particular, $v\mapsto V_T(\Gamma(v,\phi))$ is affine.  
  We have
   \begin{align}
     &\big\{ (a,b)\in[0,\infty)^2 : a1_A+b1_B\ 
    \text{can be integer replicated} \big\} \label{eq:a b} \\
    &\qquad \qquad \subseteq\big\{ \big(\delta_A(V_T(\bar\phi)),\delta_B(V_T(\bar\phi))\big): \bar\phi\in\mathcal Z \big\} \notag \\
    &\qquad \qquad =\bigcup_{\phi\in \mathcal Z_{\mathrm{c}}}\big\{ \big(\delta_A(V_T(\Gamma(v,\phi))),\delta_B(V_T(\Gamma(v,\phi)))\big): v\in\mathbb R \big\}. \label{eq:union}
   \end{align}
For each $\phi\in \mathcal Z_{\mathrm{c}}$, the set
$\big\{ \big(\delta_A(V_T(\Gamma(v,\phi))),\delta_B(V_T(\Gamma(v,\phi)))\big): v\in\mathbb R \big\}$ is a null set for the two-dimensional Lebesgue measure, because it is a one dimensional affine space in $\mathbb R^2$. We conclude that~\eqref{eq:union} has Lebesgue measure zero, and hence~\eqref{eq:a b} is a null set, too. This contradicts our assumption.
\end{proof}

%Let $\mm$ denote the set of absolutely continuous
%martingale measures that satisfy condition~(ii) of Theorem~\ref{t:real FTAP} and $\mm^{\mathrm{max}}$ %denote the set of absolutely continuous
%martingale measures that satisfy condition~(i) of Theorem~\ref{t:real FTAP}.

Recall that in the classical case (assuming {\bf (F)},
so that integrability holds), the set of arbitrage free prices has the representation
\[
  \Pi(C) = \Big\{ E_Q\Big[\frac{C}{(1+r)^T}\Big]: Q\in\mathfrak{Q} \Big\},
\]
where $\mathfrak{Q}$ is the set of equivalent martingale measures. The corresponding
result for {\bf NIA} looks as follows:

\begin{proposition}\label{prop:repr Pi}
  Assume {\bf (F)} and that the model satisfies {\bf NIA}. Let~$C$ be a claim. Then
  \[
    \Pi_{\mathbb{Z}}(C) \subseteq
    \Big\{ E_Q\Big[\frac{C}{(1+r)^T}\Big] : Q \in \mm \Big \}.
  \]
\end{proposition}\label{p:NIA preis mit MM}
\begin{proof}
  Suppose that $p\in \Pi_{\mathbb{Z}}(C)$. Then there is an adapted process $X$ such that $X_0 = p$, $X_T=C$ and $(S^0,\dots,S^d,X)$ satisfies {\bf NIA}. Let $\mmhat$ be the set of absolutely continuous
martingale measures that satisfy part~(iii) of Definition~\ref{def:Q} for this market.
By Theorem~\ref{t:real FTAP} there is $Q\in \mmhat\subseteq \mm$. Then $p = E_Q[\frac{C}{(1+r)^T}]$.
\end{proof}

The following example shows that the inclusion in Proposition \ref{p:NIA preis mit MM} can be strict. In fact, in this example we have $\Pi_{\mathbb{Z}}(C)=\emptyset$.
\begin{example}
  Let $\Omega = \{\omega_1,\omega_2,\omega_3\}$, $r=0$, $d=2$, $T=1$ and $P[\{\omega_j\}]=1/3$ for any $j=1,2,3$. Then the riskless asset is constant $1$, i.e.\ $S_t^0 = 1$ for $t\in\{0,1\}$. We choose $(S_0^1,S_0^2) = (\pi,1) $ and 
   $$ (S_1^1,S_1^2)(\omega_j) = \begin{cases}(2\pi,2)&j=1,\\(\frac{\pi}{2},\frac12)&j=2,\\(\pi,2)&j=3.\end{cases} $$
  A short calculation reveals that $Q[\{\omega_j\}] := 1_{\{j\neq 3\}}j/3$ is the only martingale measure. Obviously, we have $Q\ll P$, and the only integer strategy with zero initial wealth and $Q$-a.s.\ zero final wealth is identically zero. Thus $Q$ satisfies~(ii) in Theorem~\ref{t:real FTAP} and, hence, we have {\bf NIA}. Since $Q$ is the only martingale measure we have $\mm=\{Q\}=\mmm$. 
  
  Now, we consider the claim $C:=1_{\{\omega_3\}}$.
   %$$ C(\omega_j) := \begin{cases}0&j=1,\\0&j=2,\\1&j=3.\end{cases}$$
 Proposition \ref{p:NIA preis mit MM} yields that 
  $$\Pi_\mathbb{Z}(C) \subseteq \big\{ E_Q[C] \big\} = \{0\}. $$
  Define the extended model $(S^0,S^1,S^2,X)$, where $X_0:=0$, $X_1:=C$. Since $(0,0,0,1)$
  is an integer arbitrage for the extended model, it follows that $0\notin\Pi_\mathbb{Z}(C)$, and so
 $$ \Pi_\mathbb{Z}(C) = \emptyset \subsetneq \{0\} = \Big\{ E_R\Big[\frac{C}{(1+r)^T}\Big] : R \in \mm \Big \}.$$
\end{example}

If the model satisfies {\bf NA} (and not just {\bf NIA}), then we can compare
the sets of classical resp.\ integer arbitrage free prices. It is well known that~$\Pi(C)$
is an interval, which is open for non-replicable~$C$ and consists of a single point
if~$C$ is replicable. It turns out that under {\bf NA} the set $\Pi_\mathbb{Z}(C)$ is an interval, too, which
may differ from~$\Pi(C)$ only at the endpoints. In particular, if {\bf NA} holds,
then $\Pi_{\mathbb{Z}}(C)$ cannot be empty.

\begin{theorem}\label{thm:int inclusion}
  Suppose {\bf (F)}, that the model satisfies {\bf NA} and let~$C$ be a claim. Then
   \[
     \Pi(C) \subseteq \Pi_{\mathbb{Z}}(C) \subseteq \mathrm{cl}(\Pi(C)).
   \]
   Moreover, if $\sup(\Pi_{\mathbb{Z}}(C)) \in \Pi_{\mathbb{Z}}(C)$, then either $C$ has a
   duplication strategy in~$\mathcal Q$
   or there is no cheapest classical superhedging strategy that is in~$\mathcal Q$.
\end{theorem}
\begin{proof}
  The first inclusion is trivial. Proposition \ref{p:Qmax} in combination with {\bf NA} yields that $\mm^{\mathrm{max}}$ is the set of martingale measures which are equivalent to $P$ and that this set is dense in $\mm$. Thus, Proposition \ref{prop:repr Pi} implies
  \begin{align*}
    \Pi_{\mathbb Z}(C) &\subseteq \Big\{ E_Q\Big[\frac{C}{(1+r)^T}\Big]:Q\in\mm \Big\} \\
           &\subseteq \mathrm{cl}\Big\{ E_Q\Big[\frac{C}{(1+r)^T}\Big]:Q\in\mm^{\mathrm{max}} \Big\} \\
           &= \mathrm{cl}(\Pi(C)).
  \end{align*}
  To show the second assertion, suppose that $s:=\sup(\Pi_{\mathbb{Z}}(C)) \in \Pi_{\mathbb{Z}}(C)$,
  and that there is a cheapest classical superhedging strategy $\bar \phi\in\mathcal Q$. This means that $\bar \phi$ has price~$V_0(\bar\phi)=s$
  and payoff $V_T(\bar \phi) \geq C$. Since $s \in \Pi_{\mathbb{Z}}(C)$,
  there is an integer-arbitrage free extension of the model where~$C$
  trades at price~$s$.  Consider the strategy $(\bar \phi,-1)$ in the extended model.
  Its cost is zero, and its payoff is $V_T(\bar \phi)-C\geq0$.
  By part~(i) of Lemma~\ref{le:easy}, we conclude $C=V_T(\bar \phi)$,
  and so $\bar \phi\in\mathcal Q$ is a duplication strategy for~$C$.
\end{proof}

Alternatively, the inclusion $\Pi_{\mathbb{Z}}(C) \subseteq \mathrm{cl}(\Pi(C))$ can be proved
using Lemma~\ref{le:easy}~(i), Lemma~\ref{l:approximation}~(i), and Theorem~\ref{thm:super}.

In the preceding theorem the interval boundaries may or may not be contained in $\Pi_{\mathbb Z}(C)$,
as the following example shows. The computations needed for parts (ii)-(iv)
are similar to~(i), and we omit the details.
\begin{example}
  Let $\Omega=\{\omega_1,\omega_2,\omega_3\}$, $r=0$, $T=1$ and assume that the number of risky assets is $d=1$. Let $S_0^1 = 2$ and
   $$ S_1^1(\omega_j) = \begin{cases}1 & j=1, \\ 3 & j=2,\\ 3&j=3.\end{cases}$$
   The equivalent martingale measures are given by
   \begin{equation}\label{eq:ex mm}
     Q_\alpha := \tfrac12 \delta_{\omega_1} + \alpha \delta_{\omega_2}
      + (\tfrac12 -\alpha)\delta_{\omega_3}, \quad \alpha\in(0,\tfrac12),
   \end{equation}
   and so the model satisfies {\bf NA}.
\begin{enumerate}
  \item[(i)] Define the claim
  \[
    C(\omega_j) = \begin{cases}2\sqrt{2} & j=1, \\ 0 & j=2,\\ 4\sqrt{2}&j=3.\end{cases}
  \]
  Using~\eqref{eq:ex mm}, we find the classical arbitrage free prices
  \[
    \Pi(C)=\{ E_{Q_\alpha}[C]: \alpha \in(0,\tfrac12)\}=(\sqrt{2},3\sqrt{2}).
  \]
  We now check the boundary points for integer arbitrage, using part~(iii) of Lemma~\ref{le:easy}.
  An integer arbitrage in the market extended by~$C$ with price~$p$ thus amounts to
  $\phi\in\mathbb{Z}^2$ such that $\phi(\Delta S_1^1,C-p)\geq 0$
  and $\phi(\Delta S_1^1,C-p)\neq 0$. For $p=\sqrt 2$, we get the inequalities
  \begin{align*}
    \left(
    \begin{matrix}
      -1 & \sqrt 2 \\
      1 & -\sqrt 2 \\
      1 & 3\sqrt 2
    \end{matrix}
    \right)
    \left(
    \begin{matrix}
      \phi^1 \\
      \phi^2
    \end{matrix}
    \right)
    \geq0.
  \end{align*}
  The solution set $\{(\phi^1,\phi^1/\sqrt 2): \phi^1 \in [0,\infty)\}$ has trivial
  intersection with~$\mathbb{Z}^2$, and so $\sqrt 2$ is an integer arbitrage free
  price for~$C$. Similarly, we obtain that $3\sqrt{2} \in \Pi_{\mathbb Z}(C)$ as well,
  and we conclude that the interval $\Pi_{\mathbb Z}(C)$ contains both endpoints: $\Pi_{\mathbb Z}(C) = [\sqrt{2},3\sqrt{2}]$. 
  
  We now verify that there is no cheapest classical superhedge in~$\mathcal Q$,
  in accordance with the second assertion of Theorem~\ref{thm:int inclusion}.
  (Note that~$C$ is not replicable, as $|\Pi(C)|>1$; in particular, there is no replication strategy
  in~$\mathcal Q$.)
  Clearly, if $\bar \phi\in\mathbb{R}^2$ is a cheapest superhedge, then~$\phi^0$
  must satisfy $\phi^0=\max_{\omega\in\Omega}(C(\omega)-\phi^1 S_1^1(\omega))$.
  The cost of this strategy then is
  \begin{align*}
    V_0(\bar \phi) &= \max_{\omega\in\Omega}\big(C(\omega)-\phi^1 S_1^1(\omega)\big) + \phi^1 S_0^1\\
    &= \max_{\omega\in\Omega}\big(C(\omega)-\phi^1 \Delta S_1^1(\omega)\big).
  \end{align*}
  Our optimal strategy is $\bar \phi =(3\sqrt{2},\sqrt 2)\notin\mathcal Q$, because
  the problem
  \begin{equation*}
    \inf_{\phi^1 \in\mathbb R}\max_{\omega\in\Omega}\big(C(\omega)-\phi^1 \Delta S_1^1(\omega)\big)
    =\inf_{\phi^1 \in\mathbb R} \max\{2\sqrt 2+\phi^1,-\phi^1,4\sqrt{2}-\phi^1 \}
  \end{equation*}
  has the unique solution $\phi^1=\sqrt 2$. Similarly, we obtain that the most expensive
  classical subhedge is not in~$\mathcal Q$, agreeing with the (obvious) subhedging variant
  of the second assertion of Theorem~\ref{thm:int inclusion}.
  \item[(ii)] If \[C(\omega_j) = \begin{cases}2\sqrt{2} & j=1, \\ 0 & j=2,\\ 2\sqrt{2}&j=3,\end{cases}\]
  then $\Pi_\mathbb Z(C) = [\sqrt{2},2\sqrt{2})$. The cheapest classical superhedge
  is in~$\mathcal Q$, whereas the most expensive classical subhedge is not in~$\mathcal Q$.
  \item[(iii)] If \[C(\omega_j) = \begin{cases}0 & j=1, \\ 0 & j=2,\\ 2\sqrt{2}&j=3,\end{cases}\]
  then $\Pi_\mathbb Z(C) = (0,\sqrt{2}]$.
   The cheapest classical superhedge
  is not in~$\mathcal Q$, whereas the most expensive classical subhedge is in~$\mathcal Q$.
  \item[(iv)] If \[C(\omega_j) = \begin{cases}0 & j=1, \\ 0 & j=2,\\ 2&j=3,\end{cases}\]
  then $\Pi_\mathbb Z(C) = (0,1)$.
  The cheapest classical superhedge and the most expensive classical subhedge  are both in~$\mathcal Q$.
\end{enumerate}   
\end{example}

It might make sense to restrict attention to \emph{static} trading strategies in the claim, e.g.,
as a simple approach for modelling
the typically reduced liquidity of derivatives compared to their underlyings. This means
that the claim can initially be bought or sold, but not traded until maturity.
In the classical case, the superhedging theorem (Theorem~\ref{thm:super}) readily yields
that the set $\Pi^{\mathrm{stat}}(C)$ of static-arbitrage-free claim prices defined in this way satisfies
$\Pi^{\mathrm{stat}}(C)= \Pi(C)$. Now suppose that our model satisfies only {\bf NIA}.
Analogously to~\eqref{eq:def super}, define
\[
    \hat{\sigma}_{\mathbb{Z}}(C) = \sup \{V_0(\bar \phi): \bar\phi\in\mathcal Z,\ V_T(\bar \phi)\leq C\}.
\]
For $p\notin [ \hat{\sigma}_{\mathbb{Z}}(C), \sigma_{\mathbb{Z}}(C)]$, we clearly have
$p\notin\Pi_{\mathbb Z}^{\mathrm{stat}}(C)$, because, using appropriate integer sub- resp.\ superhedges,
one can easily construct a static integer arbitrage for the extended model.
Therefore, we obtain $\Pi_{\mathbb Z}(C)
\subseteq\Pi_{\mathbb Z}^{\mathrm{stat}}(C)\subseteq[ \hat{\sigma}_{\mathbb{Z}}(C), \sigma_{\mathbb{Z}}(C)].$

% OLD:
%Sometimes, only static trading strategies in the claim are possible, i.e.\ the claim can initially bought at any integer quantity but not bought or sold until maturity. The prices that are compatible with static trading are precisely given by the interval from the 'most expensive' subhedge to the 'cheapest' superhedge because any price in the interior of the interval cannot lead to arbitrage unless there is a sub- or superhedge for that price which there is none. Prices outside the closure of the interval cannot be compatible with static trading as the cheapest superhedge resp.\ most expensive subhedge can easily be used to build an arbitrage. 

We now proceed to identify the value of the `cheapest' superhedge in~$\mathcal Q$. The only difference to the classical case is that the cheapest superhedge is not necessarily in~$\mathcal Q$, but can be approximated arbitrarily well with superhedges in~$\mathcal Q$ (even if only {\bf NIA} holds). For results on the `cheapest' superhedge in $\mathcal Z$, see Section~\ref{a:integer superhedging}.
\begin{proposition}\label{p: superhedge}
  Assume {\bf (F)}, that the model satisfies {\bf NIA} and let $C$ be a claim. Then there is a cheapest superhedge $\bar\phi\in \mathcal R$ which satisfies
\begin{align*} V_0(\bar\phi) &= \sup\Big\{E_Q\Big[\frac{C}{(1+r)^T}\Big]:Q\in\mm\Big\} \\
 &= \sup\Big\{E_Q\Big[\frac{C}{(1+r)^T}\Big]:Q\in\mm^{\mathrm{max}}\Big\}.
 \end{align*}
   Moreover, for any $\epsilon>0$ there is a superhedge $\bar\xi\in\mathcal Q$ for $C$ such that $V_0(\bar\xi) \leq V_0(\bar\phi) + \epsilon$.
\end{proposition}
\begin{proof}
   Proposition \ref{p:Qmax} yields 
   $$\sup\Big\{E_Q\Big[\frac{C}{(1+r)^T}\Big]:Q\in\mm\Big\} = \sup\Big\{E_Q\Big[\frac{C}{(1+r)^T}\Big]:Q\in\mm^{\mathrm{max}}\Big\}.$$
  Let $Q\in\mm^{\mathrm{max}}$ and define
  \[
    A:=\{\omega\in\Omega: Q[\{\omega\}]=0\}.
  \]
   Then, $\bar S$ restricted to $\Omega\setminus A$ satisfies {\bf NA} because $Q$ is a martingale measure. By {\bf (F)}, there is a cheapest superhedge for $C$ on this market, which we denote by $\bar\eta\in\mathcal R$. It satisfies $V_T(\bar\eta) \geq C$ $Q$-a.s. Since $Q\in\mm^{\mathrm{max}}$ there is $\bar\psi\in \mathcal R_0$ such that $V_T(\bar\psi)\geq 0$ and $\{V_T(\bar\psi)>0\} = A$. As $\Omega$ is finite there is $a\in\mathbb R$ such that $V_T(a\bar\psi+\bar\eta) = aV_T(\bar\psi)+V_T(\bar\eta) \geq C$. By Proposition~\ref{p:Qmax} and
   Theorem~\ref{thm:super},
   we find that $\bar\phi := a\bar\psi+\bar\eta$ is a superhedge for $C$ with initial price 
   %\margincomment{Is $\mm^{\mathrm{max}}$ the set of emms for the model on $\Omega\setminus A$?}
   \[
     V_0(\bar\phi) = V_0(\bar\eta) = \sup\Big\{E_Q\Big[\frac{C}{(1+r)^T}\Big]:Q\in\mm^{\mathrm{max}}\Big\}.
   \]
  Now let $\bar\gamma\in\mathcal R$ be any superhedge for $C$. Then $V_T(\bar\gamma) \geq C$ $\ Q$-a.s.\ for any $Q\in\mm^{\mathrm{max}}$ and, hence, 
  $$ V_0(\bar\gamma) =  E_Q\Big[\frac{V_T(\bar\gamma)}{(1+r)^T}\Big] \geq \sup\Big\{E_Q\Big[\frac{C}{(1+r)^T}\Big]:Q\in\mm^{\mathrm{max}}\Big\}. $$
  Consequently, $\bar\phi$ is a cheapest superhedge for $C$.
  
  The second assertion follows easily from part~(i) of Lemma~\ref{l:approximation}.
\end{proof}

\section{The structure of the set of integer-arbitrage-free prices}\label{se:struct}

%We saw in Theorem~\ref{thm:int inclusion} that, under {\bf NA}, the set
%$\Pi_{\mathbb{Z}}(C)$ is an interval. The following result gives sufficient
%conditions for this property if we assume only {\bf NIA}.

The main result of this section is that the set $\Pi_{\mathbb{Z}}(C)$ of
{\bf NIA}-compatible claim prices is always dense in an interval (Theorem~\ref{thm:dense}).
First, we give some sufficient conditions that imply that $\Pi_{\mathbb{Z}}(C)$ \emph{equals} an interval.

\begin{proposition}\label{prop:int}
  Assume {\bf (F)} and that the model satisfies {\bf NIA}.
  Then $\Pi_{\mathbb{Z}}(C)$ is an interval if any of the following statements holds:
   \begin{itemize}
     \item[(i)] there is only one trading period ($T=1$), and $\Pi_{\mathbb{Z}}(C)$ is not empty,
     \item[(ii)] there is only one risky asset ($d=1$), or
     \item[(iii)] the model satisfies {\bf NA}.
   \end{itemize}
\end{proposition}
%
%\margincomment{Proposition~\ref{prop:int} works for any probability space and any restriction set $\xi\ni\mathcal{X}\subseteq\mathbb{R}^d$}
%
\begin{proof}
  If~(ii) holds, then Theorem \ref{t:NA NIFL} yields that~(iii) holds.  
  If we assume~(iii), then Theorem~\ref{thm:int inclusion} yields the claim.  
  Now assume that~(i) holds. Proposition~\ref{prop:repr Pi} implies that
  \[
    \Pi_{\mathbb{Z}}(C) \subseteq \Big\{E_Q\Big[\frac C{(1+r)^T}\Big] : Q\in\mm\Big\} :=J.
  \]
   If $J$ is a singleton, then we have equality by assumption and, hence, the claim follows. Thus we may assume that $J$ contains at least two points. The set $J$ is an interval. Let $p\in \inn(J)$ and define $X_0:=p$, $X_1:=C$. Assume for contradiction that there is an integer arbitrage $(\bar\phi,\phi^{d+1})$ for the model $(S^0,\dots,S^d,X)$.
   By part~(ii) of Lemma~\ref{le:easy}, we may assume that $V_0(\bar \phi,\phi^{d+1})=0$.   
   We have $\phi_1^{d+1}\neq 0$, because otherwise $\bar\phi$ is an integer arbitrage for $(S^0,\dots,S^d)$ with $V_0(\bar\phi) = 0$. 
   
   {\em Case 1:} $\phi_1^{d+1}<0$. Then $C \leq -(1/\phi_1^{d+1})\sum_{j=0}^d \phi_1^j S_1^j$. Thus, 
    $$ \psi_1^j := -\frac{\phi^j}{\phi^{d+1}},\quad j=0,\dots,d, $$
 is a superhedge for $C$. We have $V_0(\bar\psi) = p$. Proposition \ref{p: superhedge} yields that $p = V_0(\bar\psi) \geq \sup(J) > p$. A contradiction.
  
   {\em Case 2:} $\phi_1^{d+1}>0;$ analogous.
 
  Thus, $p\in\Pi_{\mathbb{Z}}(C)$ which yields that $\inn(J) \subseteq  \Pi_{\mathbb{Z}}(C) \subseteq J$ and, hence, $\Pi_{\mathbb{Z}}(C)$ is an interval.
\end{proof}

 We now give an example in which the set of integer arbitrage compatible prices is \emph{not} an interval. More precisely, we exhibit a model satisfying {\bf NIA} and a claim $C$ where the set of {\bf NIA}-consistent prices is given by $\Pi_{\mathbb{Z}}(C) = [0,1/2]\setminus (\mathbb Q + \mathbb Q \pi)$.
\begin{example}
  Let $\Omega := \{\omega_1,\omega_2,\omega_3,\omega_4\}$, $d=2$, $T=2$ and $r=0$. We use the filtration $\mathcal F_0 := \{\emptyset,\Omega\}$, $\mathcal F_1:=\sigma(\{\omega_1\},\{\omega_2,\omega_3\},\{\omega_4\})$ and $\mathcal F_2:=2^\Omega$. We choose the market model given by $(S^1_0,S^2_0) := (1,\pi)$ and
  $$ S_2(\omega_j) := S_1(\omega_j) := \begin{cases} (3/2,3\pi/2) & j=1,\\ (1/2,\pi/2) & j=2,3, \\ (1,1+\pi) & j=4. \end{cases}$$
  This market allows for the static \emph{real} arbitrage $\bar\eta_t := (0,-\pi,1)$ which is self-financing, satisfies $V_0(\bar\eta) = 0$ and $V_2(\bar\eta) = 1_{\{\omega_4\}}$. Thus, any martingale measure $Q$ must satisfy $Q[\{\omega_4\}] = 0$. Define the measure $Q_\alpha$ by
   $$Q_\alpha[\{\omega_j\}] := \begin{cases} 1/2&j=1,\\ \alpha/2 &j=2,\\ (1-\alpha)/2& j=3,\\ 0 & j=4. \end{cases}$$
   for any $\alpha \in [0,1]$. Then $Q_\alpha$ is a martingale measure and $\mathcal Z_{Q_\alpha}^0 = \{0\}$. In particular, Theorem \ref{t:real FTAP} yields that {\bf NIA} holds. Moreover, $\mm = \{Q_\alpha: \alpha\in[0,1]\}$.
   
   Now we choose the claim $C := 1_{\{\omega_2\}}$. Proposition~\ref{prop:repr Pi} yields
    $$\Pi_{\mathbb{Z}}(C) \subseteq \{ E_Q[C] : Q\in\mm \} = [0,1/2].$$

%For $\alpha\in[0,1]$ we define $X^\alpha_0 := \alpha/2$, $X^\alpha_1 := \alpha 1_{\omega_2,\omega_3}$, $X^\alpha_2 := C$. 

 Let $p\in [0,1/2]\cap (\mathbb Q+ \mathbb Q \pi)$, $p\neq 0$ and assume for contradiction that $p\in \Pi_{\mathbb{Z}}(C)$. Then there is an adapted process $(X_0,X_1,X_2)$ such that $X_0=p$, $X_2=C$ and the model $(S^0,S^1,S^2,X)$satisfies {\bf NIA}. Define $\alpha := 2p$, and let $u,v\in\mathbb Q$
 be such that $\alpha=u+v\pi$.
 %By Proposition~\ref{prop:repr Pi}, there is a martingale measure $Q\in\mm$ such that $\alpha/2 = p = E_Q[C]$. The only martingale measure with this property is $Q_\alpha$. We get $X_1=E_{Q_\alpha}[C|\mathcal{F}_1] = \alpha 1_{\{\omega_2,\omega_3\}}$.
 Define the strategy $(\bar \phi_t,\phi^3_t)_{t=1,2} \in \mathcal Q$ by
 \[
   \phi^3_1:= \mathrm{sgn}((2-\pi)v-u) \in \{-1,1\}
 \]
 and
 \begin{align*}
   (\bar \phi_1,\phi_1^3) &:= (-\tfrac32 \alpha \phi_1^3,u \phi_1^3,v \phi_1^3,\phi^3_1), \\
   (\bar \phi_2,\phi_2^3)(\omega_j) &:=
   \begin{cases}
     0 & j=1,2,3, \\
     \big(\tfrac12|(2-\pi)v-u|,0,0,0\big) & j=4.
   \end{cases}
 \end{align*}
 This strategy satisfies $V_0(\bar\phi,\phi^3) = 0$ and
 \[
   V_1(\bar\phi,\phi^3)=V_2(\bar\phi,\phi^3) = \tfrac12|(2-\pi)v-u|\cdot 1_{\{\omega_4\}}.
 \]
   Thus, $(\bar\phi,\phi^3)$ is a rational arbitrage. A contradiction.
 
Let $p=0$ and assume for contradiction that $p\in \Pi_{\mathbb{Z}}(C)$. Then there is an adapted process $(X_0,X_1,X_2)$ such that $X_0=0$, $X_2=C$ and the model $(S^0,S^1,S^2,X)$ satisfies {\bf NIA}. Define the static strategy $(\bar\phi,\phi^3) := (0,0,0,1)$. We have $V_0(\bar\phi,\phi^3) = 0$ and $V_2(\bar\phi,\phi^3) = 1_{\{\omega_2\}}$. Thus, $(\bar\phi,\phi^3)$ is an integer arbitrage. A contradiction.
 
 We have shown so far that $\Pi_{\mathbb{Z}}(C)\subseteq [0,1/2]\setminus (\mathbb Q + \mathbb Q \pi)$.
 Conversely, let now $p\in [0,1/2]\setminus (\mathbb Q + \mathbb Q\pi)$. We show that $p\in \Pi_{\mathbb{Z}}(C)$. Define $\alpha := 2p$ and  $X^\alpha_0 := \alpha/2$, $X^\alpha_1 := \alpha 1_{\{\omega_2,\omega_3\}}$, $X^\alpha_2 := C$. To see that the model $(S^0,S^1,S^2,X^\alpha)$ satisfies {\bf NIA},
assume for contradiction that there is an integer arbitrage. Then there is a one period arbitrage $(\bar\phi,\phi^3)$. Obviously, there is no arbitrage possibility in the second period, and so we may assume $(\phi_2,\phi^3_2)=0$ (i.e., no risky position in the second period). Thus, $V_1(\bar\phi,\phi^3) = V_2(\bar\phi,\phi^3) \geq 0$. Since $Q_\alpha$ is a martingale measure for the extended model, we get $V_1(\bar\phi,\phi^3) = 0$ $Q_\alpha$-a.s.
In particular, $V_1(\bar\phi,\phi^3)(\omega_2)=0$, and together with $V_0(\bar\phi,\phi^3)=0$
(see Lemma~\ref{le:easy}~(ii)) this implies
\[
  \phi^1_1 + \pi \phi^2_1 - \alpha \phi^3_1=0.
\]
As the original model satisfies {\bf NIA}, we must have $\phi^3_1\neq 0$, which leads to the
contradiction
\[
  2p=\alpha = \frac{\phi^1_1 + \pi \phi^2_1}{\phi^3_1} \in \mathbb Q + \mathbb Q\pi.
\]
 Thus, there is no integer arbitrage, i.e.\ the model satisfies {\bf NIA} and, hence, $p\in \Pi_{\mathbb{Z}}(C)$.
\end{example}

Throughout the remainder of this section, we will always assume {\bf (F)}, {\bf NIA} and $\mathcal F_T=\mathcal A$. Also, let $C$ be a claim. The following theorem is our main result on the structure of $\Pi_\mathbb Z(C)$ in the general case.
\begin{theorem}\label{thm:dense}
  The set $\Pi_\mathbb Z(C)$ is either empty or dense in
   $$ \left[\inf\left\{ E_Q\left[\frac C{(1+r)^T}\right]: Q\in\mmm\right\},\sup\left\{ E_Q\left[\frac C{(1+r)^T}\right]: Q\in\mmm\right\} \right]. $$
\end{theorem}

The theorem will follow from Lemmas~\ref{l:c0 good} and~\ref{l:countable} below.
Note that~$\mmm$ can be replaced by~$\mm$, due to the last assertion of
Proposition~\ref{p:Qmax}.
In order to prove Theorem~\ref{thm:dense}, we assume that $\Pi_\mathbb Z(C)$ is non-empty, and choose $p^*\in \Pi_\mathbb Z(C)$. By definition, there is an adapted process $(X^*_t)_{t\in\mathbb T}$ such that $X^*_0=p$, $X^*_T=C$, and the model $(S^0,\dots,S^d,X^*)$ satisfies {\bf NIA}. We also define
$$A_t := \{ \omega\in\Omega: \exists \bar\phi\in\mathcal R_0:V_{T}(\bar\phi)\geq 0,V_{T}(\bar\phi)(\omega)>0, \bar\phi_1,\dots,\bar\phi_{t}=0\} $$
for any $t\in\mathbb T\setminus\{T\}$. By the same argument as in the proof of Theorem \ref{t:real FTAP}, there is $\bar\phi\in\mathcal R_0$ with $\bar\phi_1,\dots,\bar\phi_t=0$, $V_{T}(\bar\phi)\geq 0$ and $\{V_{T}(\bar\phi)> 0\} = A_t$. From the definition we see that $A_{t} \subseteq A_{t-1}$ and $A_{t-1}\setminus A_t\in\mathcal F_t$ for any $t\in\mathbb T\setminus\{0\}$.

\begin{definition}
We write $\mmt$ for the set of measures $Q$ such that $(\hat S_u)_{u=t,\dots T}$ is a $Q$-martingale, $\Omega\setminus A_t$ is the support of~$Q$, and $Q[B]>0$ for any non-empty set $B\in\mathcal F_t$.
Now we define two sequences of sets:
\begin{align*}
  \mathcal K_T &:= \{C\}, \\
  \mathcal K_{t} &:= \Big\{E_Q\left[\frac {D}{1+r}\Big|\mathcal F_t\right]: Q\in\mmt,D\in\mathcal K_{t+1} \Big\}, \\
  \mathcal C_T &:= \{C\}, \\
  \mathcal C_{t} &:= \Bigg\{E_Q\left[\frac {D}{1+r}\Big|\mathcal F_t\right]: Q\in\mmt,D\in\mathcal C_{t+1}, \forall B\in\mathcal F_t\, \forall \xi\in \mathbb Q^d\, \forall s\in\{-1,1\}:\\
  &\quad\quad 1_B\, \xi\Delta\hat S_{t+1} \geq s1_B\left(\frac{D}{(1+r)^{t+1}}-E_Q\Big[\frac{D}{(1+r)^{t+1}}\Big|\mathcal F_t\Big]\right) \\
  &\quad\quad\Rightarrow 1_B\, \xi\Delta\hat S_{t+1} = s1_B\left(\frac{D}{(1+r)^{t+1}}-E_Q\Big[\frac{D}{(1+r)^{t+1}}\Big|\mathcal F_t\Big]\right)  \Bigg\}
\end{align*}
for any $t\in\mathbb T\setminus\{T\}$.
\end{definition}
Lemma~\ref{l:tower property} below together with the convexity of $\mmm=\mathfrak Q_0$ implies that 
 $$ \mathcal K_0 = \left[\inf\left\{ E_Q\left[\frac C{(1+r)^T}\right]: Q\in\mmm\right\},\sup\left\{ E_Q\left[\frac C{(1+r)^T}\right]: Q\in\mmm\right\} \right], $$
and Lemma \ref{l:countable} below yields a countable exception set $F$ such that $\mathcal C_0 = \mathcal K_0\setminus F$. Finally, Lemma \ref{l:c0 good} states that $\mathcal C_0$ is contained in the set $\Pi_\mathbb Z(C)$ of {\bf NIA}-compatible prices, which establishes Theorem~\ref{thm:dense}. For technical reasons, we first analyse the sets $\mmt$, and we will need the stochastic convexity of $\mathcal K_t$ given in Lemma \ref{l:superconvex} below.

\begin{lemma}\label{l:mmt structure}
  Let $t\in \mathbb T$. Then $\mmt$ is non-empty. If $t\neq 0$, then for any $Q\in\mathfrak Q_{t-1}$ there is $Q'\in\mmt$ such that
   $$ E_{Q'}[X|\mathcal F_s] = E_Q[X|\mathcal F_s]\quad Q\text{-a.s.}$$
 for any $s=t,\dots,T$ and any random variable $X:\Omega\rightarrow\mathbb R$.
\end{lemma}
\begin{proof}
  Let $I:=\{t\in\mathbb T: \text{ the claim holds for }t\}$. We have $0\in I$ by
  Theorem~\ref{t:real FTAP}~(i). Let $t\in\mathbb T$ such that $t-1\in I$. We show $t\in I$ which implies $I=\mathbb T$ and, hence, the claim. We directly produce the measure with the given extra property. To this end let $Q\in \mathfrak Q_{t-1}$. Let $B_1\dots,B_m$ be an enumeration of the minimal non-empty elements of $\mathcal F_t$ and define $k := |\{ B_l: l=1,\dots,m, B_l\subseteq A_{t-1}\setminus A_t\}|$. We may assume that $B_1,\dots,B_k\subseteq A_{t-1}\setminus A_t$. Since $\Omega\setminus A_{t-1}$ is the support of $Q$ we have $Q[B_l] > 0$ for any $l=k+1,\dots,m$. Since $A_t\setminus A_{t-1}$ is $\mathcal F_t$-measurable we have $A_{t-1}\setminus A_t = \bigcup_{l=1}^kB_l$. Define the probability measures $P_l := \frac{P}{P[B_l]}$ on $B_l$. Since the model $(\hat S_u)_{u=t,\dots,T}$ satisfies {\bf NIA} we get from Theorem \ref{t:real FTAP} (i) a martingale measure $Q_l\ll P_l$ on $B_l$. Define the probability measure
   $$ Q'[D] := \left(Q[D] + \sum_{l=1}^k Q_l[D\cap B_l] \right)/(1+k),\quad D\in\mathcal A. $$
 Clearly, the support of $Q'$ is $\Omega\setminus A_t$ and $Q'[D]>0$ for any $D\in\mathcal F_t$. Also, $(\hat S_u)_{u=t,\dots,T}$ is a $Q'$-martingale. Thus, we have $Q'\in\mmt$. 
 
 Now, let $X:\Omega\rightarrow \mathbb R$ be a random variable and $s\in\{t,\dots,T\}$. We show that $E_{Q'}[X|\mathcal F_s]$ is a version of the $\mathcal F_s$-conditional expectation of $X$ under $Q$. To this end, let $D\in\mathcal F_s$ and define $D':=D\setminus A_{t-1}$. Then $D'$ is $Q'$-essentially $\mathcal F_s$-measurable, because $A_t$ is a $Q'$-null set and $A_{t-1}\setminus A_t$ is $\mathcal F_t\subseteq\mathcal F_s$-measurable. We have
 \begin{align*}
    E_Q\big[E_{Q'}[X|\mathcal F_s]1_D\big] &= E_Q\big[E_{Q'}[X|\mathcal F_s]1_{D'}\big] \\
                                 &= E_Q\big[E_{Q'}[X1_{D'}|\mathcal F_s]\big] \\
                                 &= E_Q\big[E_{Q}[X1_{D'}|\mathcal F_s]\big] \\
                                 &= E_Q[X1_{D'}] \\
                                 &= E_Q[X1_D].
 \end{align*}
 Thus, $t\in I$.
\end{proof}

\begin{lemma}\label{l:tower property}
  For any $t\in\mathbb T$ we have 
     $\mathcal K_{t} = \left\{ E_Q\left[\frac{C}{(1+r)^{T-t}}\Big|\mathcal F_t\right]: Q\in\mmt\right\}$.
\end{lemma}
\begin{proof}
  Define
   $$I:=\left\{ t\in\mathbb T: \mathcal K_{t} = \left\{ E_Q\left[\frac{C}{(1+r)^{T-t}}\Big|\mathcal F_t\right]: Q\in\mmt\right\}\right\}.$$
 Obviously, $T\in I$. Let $t\in I\setminus\{0\}$. We show that $t-1\in I$ which implies $I=\mathbb T$ and, hence, the claim.

  To this end, let $X_{t-1}\in \mathcal K_{t-1}$. Then there is $Q\in\mathfrak Q_{t-1}$ and $X_t\in\mathcal K_t$ such that $X_{t-1} = E_Q[\frac{X_t}{1+r}|\mathcal F_{t-1}]$. Since $X_t\in\mathcal K_t$ there is $R\in\mmt$ such that $X_t = E_R[\frac{C}{(1+r)^{T-t}}|\mathcal F_{t}]$. Define the measure
  $$ Q'[B] := E_Q\big[R[B|\mathcal F_t]\big],\quad B\in\mathcal A. $$
Since $R\in\mmt$ we have $R[B] >0$ for any non-empty set $B\in\mathcal F_t$. Let $B\in\mathcal F_{t-1}\subseteq\mathcal F_t$ be non-empty. Then $R[B|\mathcal F_t] = 1_B$ and, hence, $Q'[B]=Q[B]>0$. Also, $(\hat S_u)_{u=t-1,\dots,T}$ is a $Q'$-martingale. Since $A_t\subseteq A_{t-1}$ and $A_{t-1}\setminus A_t\in\mathcal F_t$ we get
 $$ R[A_{t-1}|\mathcal F_t] = 1_{A_{t-1}\setminus A_t} + R[A_t|\mathcal F_t]. $$
However, $A_t$ is an $R$ null set, hence $R[A_t|\mathcal F_t] = 0$ $R$-a.s. Since $\mathcal F_t$ has no non-empty $R$ null sets, we have $R[A_t|\mathcal F_t] = 0$. We get $R[A_{t-1}|\mathcal F_t] = 1_{A_{t-1}\setminus A_t},$ which yields $Q'[A_{t-1}] = Q[A_{t-1}\setminus A_t] = 0$. Let $\omega\in\Omega\setminus A_{t-1}$. Then $R[\{\omega\}|\mathcal F_t]\geq 0$ and $R[\{\omega\}|\mathcal F_t](\omega) > 0$. Since $Q[\{\omega\}]>0$ we get $Q'[\{\omega\}] >0$. Thus, the support of $Q'$ is $\Omega\setminus A_{t-1}$, which yields $Q'\in\mathfrak Q_{t-1}$. We have
 \begin{align*}
  E_{Q'}\Big[\frac{C}{(1+r)^{T-(t-1)}}\Big|\mathcal F_{t-1}\Big] &= E_{Q}\Bigg[E_R\Big[\frac{C}{(1+r)^{T-t}}\Big|\mathcal F_t\Big]/(1+r)\Bigg|\mathcal F_{t-1}\Bigg] \\
                     &= E_Q\Big[\frac{X_t}{1+r}\Big|\mathcal F_{t-1}\Big] \\
                     &= X_{t-1}.
 \end{align*}
 Thus, $X_{t-1} \in \left\{ E_Q\Big[\frac{C}{(1+r)^{T-t}}\Big|\mathcal F_{t-1}\Big]: Q\in\mathcal Q_{t-1}\right\}$.
 
 Now, let $X_{t-1}\in \left\{ E_Q[\frac{C}{(1+r)^{T-(t-1)}}|\mathcal F_{t-1}]: Q\in\mathfrak Q_{t-1}\right\}$; we have to show that $X_{t-1}\in \mathcal K_{t-1}$. There is $Q\in\mathfrak Q_{t-1}$ such that $X_{t-1} = E_Q[\frac{C}{(1+r)^{T-(t-1)}}|\mathcal F_{t-1}]$. By Lemma~\ref{l:mmt structure} we find $Q'\in\mmt$ such that $E_Q[Y|\mathcal F_s] = E_{Q'}[Y|\mathcal F_s]$ $Q$-a.s.\ for any random variable $Y:\Omega\rightarrow \mathbb R$ and any $s=t,\dots, T$. Define $X_t := E_{Q'}[\frac{C}{(1+r)^{T-t}}|\mathcal F_t] \in \mathcal K_t$ because $t\in I$. We find
\begin{align*}
 \mathcal K_{t-1} &\ni E_Q\Big[\frac{X_t}{1+r}\Big|\mathcal F_{t-1}\Big] \\
                  &= E_Q\Bigg[E_{Q'}\Big[\frac{C}{(1+r)^{T-t+1}}\Big|\mathcal F_t\Big]\Bigg|\mathcal F_{t-1}\Bigg] \\
                  &= E_Q\Big[\frac{C}{(1+r)^{T-(t-1)}}\Big|\mathcal F_{t-1}\Big] \\
                  &= X_{t-1}.
\end{align*} 
Thus, $t-1\in I$.
\end{proof}

\begin{lemma}\label{l:superconvex}
  For any $t\in\mathbb T$ and any $\mathcal F_t$-measurable random variable $\alpha$ with values in $[0,1]$ and any $X,Y\in\mathcal K_t$ we have
   $$ \alpha X + (1-\alpha) Y \in \mathcal K_t. $$
\end{lemma}
\begin{proof}
 Lemma \ref{l:tower property} yields measures $Q,R\in\mmt$ such that $X=E_Q[\frac{C}{(1+r)^{T-t}}|\mathcal F_t]$, $Y=E_R[\frac{C}{(1+r)^{T-t}}|\mathcal F_t]$. Define the measure
  $$ Q'[B] := E_Q\big[ \alpha Q[B|\mathcal F_t] + (1-\alpha) R[B|\mathcal F_t]\big]. $$
  It is clear that $Q$ and~$Q'$ agree on~$\mathcal{F}_t$, and
one easily verifies $Q'\in\mmt$. Let $B\in\mathcal{F}_t$. Then
\begin{align*}
  E_{Q'}\Big[ 1_B \frac{C}{(1+r)^{T-t}}\Big] &=
    E_Q\Big[ \alpha E_Q\Big[ 1_B \frac{C}{(1+r)^{T-t}} \Big| \mathcal{F}_t\Big]
    + (1-\alpha) E_R\Big[ 1_B \frac{C}{(1+r)^{T-t}} \Big| \mathcal{F}_t\Big]\Big] \\
  &= E_Q[\alpha 1_B X +(1-\alpha)1_B Y] \\
  &= E_{Q'}[\alpha 1_B X +(1-\alpha)1_B Y].
\end{align*}
 We find
 \begin{align*}
   \mathcal K_t &\ni E_{Q'}\Big[\frac{C}{(1+r)^{T-t}}\Big|\mathcal F_t\Big] \\
                &= \alpha X + (1-\alpha)Y.
 \end{align*}
\end{proof}

\begin{lemma}\label{l:c0 good}
 We have $\mathcal C_0 \subseteq \Pi_\mathbb Z(C)$.
\end{lemma}
\begin{proof}
  Let $p\in\mathcal C_0$. Define $X_0:=p$. We can find recursively $X_{t+1}\in\mathcal C_{t+1}$ and $Q_t\in\mmt$ such that $X_t = E_{Q_t}[X_{t+1}/(1+r)|\mathcal F_t]$ for $t\in\mathbb T\setminus\{T\}$. Since $X_T\in\mathcal C_T = \{C\}$ we have $X_T=C$. Assume for contradiction that there is an integer arbitrage for the model $(S^0,\dots,S^d,X)$. Then there is a one period integer arbitrage $(\bar\phi,\phi^{d+1})$, i.e.\ there is $t_0\in\mathbb T$ such that $(\phi_t,\phi^{d+1}_t)=0$ for any $t\in\mathbb T\setminus\{0,t_0\}$. Then there is a minimal set $B\in\mathcal F_{t_0-1}\setminus\{\emptyset\}$ such that $1_B\, (\bar\phi,\phi^{d+1})$ is still an arbitrage. Define $(\eta,\eta^{d+1}):=(\phi_{t_0},\phi^{d+1}_{t_0})(\omega)\in\mathbb Z^{d+1}$ for some $\omega\in B$. Then
  \[
    Y := 1_B\left(\eta\Delta \hat S_{t_0} + \eta^{d+1}\left(\frac{X_{t_0}}{(1+r)^{t_0}}-\frac{X_{t_0-1}}{(1+r)^{t_0-1}}\right)\right) \geq 0
  \]
and $P[Y>0]>0$. Since the model $(S^0,\dots,S^d)$ satisfies {\bf NIA}, we have
$\eta^{d+1}\neq0$ and can define $\xi^j := \eta^j/\eta^{d+1}$. We get
 $$ Y/\eta^{d+1} = 1_B\left( \xi\Delta \hat S_{t_0} + \left(\frac{X_{t_0}}{(1+r)^{t_0}}-\frac{X_{t_0-1}}{(1+r)^{t_0-1}}\right)\right).$$
 Thus, $X_{t_0-1}\notin \mathcal C_{t_0-1}$. A contradiction.
\end{proof}

It is not hard to see that actually $\mathcal C_0 = \Pi_\mathbb Z(C)$, but we will not use this fact.

\begin{lemma}\label{l:countable}
  Let $t\in\mathbb T$. Let $X_t\in\mathcal K_t$ and define $X^\alpha_t := \alpha X_t + (1-\alpha)X_t^*$ for any $\alpha\in[0,1]$. Then there is a countable set $F\subseteq (0,1]$ such that $X_t^\alpha \in \mathcal C_t$ for any $\alpha\in[0,1]\setminus F$. In particular, $\mathcal C_t$ is dense in $\mathcal K_t$.
\end{lemma}
\begin{proof}
  Define
   $$I:=\{ t\in\mathbb T: \text{ the claim holds for this }t\}.$$
 Obviously, $T\in I$. Let $t\in I\setminus\{0\}$. We show that $t-1\in I$ which implies $I=\mathbb T$ and, hence, the claim. To this end, let $X_{t-1}\in\mathcal K_{t-1}$. Then there are $X_t\in\mathcal K_t$ and $Q\in\mathfrak Q_{t-1}$ such that $X_{t-1} = E_Q[X_t/(1+r)|\mathcal F_{t-1}]$. We define
 \begin{align*}
    X_{t-1}^\alpha &:= \alpha X_{t-1} + (1-\alpha) X^*_{t-1}, \\
    X_{t}^\alpha &:= \alpha X_{t} + (1-\alpha) X^*_{t}
 \end{align*}
for any $\alpha\in[0,1]$. There is $F_t\subseteq (0,1]$ countable such that $X_t^{\alpha}\in\mathcal C_t$ for any $\alpha\in [0,1]\setminus F_t$. For any $\alpha\in[0,1]\setminus F_t$ we find recursively $X^{\alpha}_{s+1}\in\mathcal C_{s+1}$ and $Q^\alpha_s\in\mathfrak Q_s$ such that $X^\alpha_s = E_{Q^\alpha_s}[X^\alpha_{s+1}/(1+r)|\mathcal F_s]$, for $s\geq t$.
 
 We will show that $(S_u^0,\dots,S_u^d,X^\alpha_u)_{u=t-1,\dots,T}$ satisfies {\bf NIA} for all but countably many choices for $\alpha\in[0,1]$. Since existence of an integer arbitrage implies existence of a one-period arbitrage and the market $(S_u^0,\dots,S_u^d,X^\alpha_u)_{u=t,\dots,T}$ does not allow for arbitrage, we know that this arbitrage must be in the period from $t-1$ to $t$. Since $\mathcal F_{t-1}$ is generated by finitely many atoms it is sufficient to condition on one of the atoms. Thus, we may simply assume that $t=1$.
We define
  $$ F_{0} := \{ \alpha\in(0,1] : \alpha\in F_1\text{ or } X^\alpha_{0} \notin \mathcal C_0\}, $$
 and we will show that $F_{0}$ is countable. To this end, we define the sets
  \begin{align*}
     \mathcal D_1 &:= \{ Y\in L^0((\Omega,\mathcal F_{1},P),\mathbb R): Y=\xi\Delta\hat S_1\ Q\text{-a.s.}, \xi \in \mathbb R^d\}, \\
     \mathcal D^{\mathbb Q}_1 &:= \{ Y\in L^0((\Omega,\mathcal F_{1},P),\mathbb R): Y=\xi\Delta\hat S_1\ Q\text{-a.s.}, \xi \in \mathbb Q^d\}
  \end{align*}
and $\Delta \hat X^\alpha_1 := \frac{X^\alpha_1}{1+r}-X^\alpha_0$ for $\alpha\in[0,1]$.

 {\em Case 1}: $\Delta \hat X_1\notin \mathcal D_1$ or $\Delta \hat X^*_1\notin \mathcal D_1$.  We define $\bar F:=\{\alpha\in [0,1]:\Delta X_1^\alpha\in \mathcal D_1\}$. Since $\mathcal D_1$ is a vector space, we find that $\bar F$ contains at most one element. We claim that $(0,1]\setminus (F_1\cup\bar F)$ does not contain any element of $F_0$. To this end, let $\alpha \in (0,1]\setminus (F_1\cup\bar F)$ and assume for contradiction that there is an integer arbitrage in the first period. Hence, there is $\xi\in\mathbb Z^{d+1}$ such that $\xi^{d+1}\neq 0$ and $\xi\Delta\hat S_1 + \xi^{d+1}\Delta \hat X^\alpha_1\geq 0$. Since $Q$ is a martingale measure we get $\xi\Delta\hat S_1 + \xi^{d+1}\Delta \hat X^\alpha_1 = 0$ $Q$-a.s., and after solving for $\Delta \hat X^\alpha_1$ we find $\Delta \hat X^\alpha_1\in\mathcal D_1$. A contradiction.
 
 {\em Case 2}: $\Delta \hat X_1,\Delta \hat X^*_1 \in\mathcal D_1$ and there is a set with positive $Q$-measure on which $\Delta \hat X_1 \neq \Delta \hat X^*_1$. Since $\mathcal D^{\mathbb Q}_1$ restricted to $\Omega\setminus A_0$ has countably many elements we find that $\Delta\hat X^\alpha_1\in \mathcal D^{\mathbb Q}_1$ at most countably often. Denote the set of $\alpha\in[0,1]$ where $\Delta\hat X^\alpha_1\in \mathcal D^{\mathbb Q}_1$ by $\bar F$. We claim $F_0\subseteq F_1\cup \bar F$. To this end, let $\alpha \in (0,1]\setminus (F_1\cup \bar F)$ and assume for contradiction that there is an integer arbitrage in the first period. Then there is $\xi\in\mathbb Z^{d+1}$ such that $\xi^{d+1}\neq 0$ and $\xi\Delta\hat S_1 + \xi^{d+1}\Delta \hat X^\alpha_1\geq 0$. Since $Q$ is a martingale measure we get $\xi\Delta\hat S_1 + \xi^{d+1}\Delta \hat X^\alpha_1 = 0$ $Q$-a.s., and after solving for $\Delta \hat X^\alpha_1$ we find $\Delta \hat X^\alpha_1\in\mathcal D_1^{\mathbb Q}$. A contradiction.
 
 {\em Case 3}: $\Delta \hat X_1,\Delta \hat X^*_1 \in\mathcal D_1$ and $\Delta\hat X_1 = \Delta \hat X^*_1$ $Q$-a.s. Then, we have
 \begin{align}
   X_1 &= (1+r)(\Delta\hat X_1+X_0) \notag \\
    &= (1+r)(\Delta\hat X^*_1+X_0) \notag \\
     &= X_1^* + (1+r)(X_0-X_0^*),\quad Q\text{-a.s.} \label{eq:XX^*}
 \end{align}
 If $X_0 = X_0^*$, then $X_0=X_0^*\in \mathcal C_0$. Thus we may assume that $X_0\neq X_0^*$. If $R[B] \in \{0,1\}$ for any $R\in\mmm$ and any $B\in\mathcal F_1$, then Lemma \ref{l:two nodes} yields $\mathcal F_1=\mathcal F_0$ and, hence, $\Delta \hat X_1 = 0 = \Delta \hat X^*_1$ which yields that $F_0 = F_1$. Thus, we may assume that there is $R\in\mmm$ and $B\in\mathcal F_1$ with $R[B]\in (0,1)$. By Proposition \ref{p:Qmax} we find that there is $B\in\mathcal F_1$ such that for \emph{any} $R\in\mmm=\mathfrak{Q}_0$ we have $R[B]\in (0,1)$. In particular, we have $Q[B]\in (0,1)$. %Since the model $(S^0,\dots,S^d,X^*)$ satisfies {\bf NIA} by assumption, Theorem \ref{t:real FTAP} yields $Q^*\in\mm$ such that $(S^0,\dots,S^d,X^*)$ is a $Q^*$-martingale.
 
For $n\in\mathbb N$ we define
  \begin{align*} 
     Y_1^n &:= X_11_{B^{\mathrm c}} + \big((1-1/n)X_1+X_1^*/n\big)1_{B} \\
           &= X_1 + 1_B(X_1^*-X_1)/n, \\
     Y_0^n &:= X_0Q[B^{\mathrm c}] + \big((1-1/n)X_0+X_0^*/n\big)Q[B] \\
           &= X_0 + Q[B](X_0^*-X_0)/n.
  \end{align*}
 Lemma~\ref{l:superconvex} yields that $Y_1^n\in\mathcal K_1$ for any $n\in\mathbb N$. The measure $Q_n[D] := E_Q[Q_n'[D|\mathcal F_1]]$, where $Q_n'\in\mathfrak{Q}_1$ is such that $Y_1^n = E_{Q_n'}[\frac{C}{(1+r)^{T-1}}|\mathcal{F}_1]$, satisfies
 \[
  E_{Q_n}\Big[\frac{Y_1^n}{1+r}\Big] = E_{Q}\Big[\frac{Y_1^n}{1+r}\Big] = Y_0^n,
  \]
  where the last equality follows from the definition of~$Q$ and~\eqref{eq:XX^*}.
  Thus, $Y_0^n\in\mathcal K_0$ for any $n\in\mathbb N$. Observe that
$$\Delta \hat Y_1^n := \frac{Y_1^n}{1+r}-Y_0^n = \Delta \hat X_1 + \frac{1}{n}\left(1_B\frac{X_1^*-X_1}{1+r}-Q[B](X_0^*-X_0)\right).$$
 We find that $\Delta \hat Y_1^n \neq \Delta \hat X_1^*$ with positive $Q$-probability. By appealing to case 1 resp.\ case 2 we find $F^n\subseteq[0,1]$ countable such that $\alpha Y_0^n+(1-\alpha)X_0^*\in\mathcal C_0$ for any $\alpha\in [0,1]\setminus F^n$. Define  the countable set
  $$\bar F:= \{\alpha(1-Q[B]/n): n\in\mathbb N, \alpha\in F^n\}.$$
We claim that $F_0\setminus\{1\}\subseteq F_1\cup \bar F$. To this end let $\alpha\in (0,1) \setminus (F_1\cup \bar F)$. Choose $n\in\mathbb N$ such that $n > Q[B]/(1-\alpha)$.
Then there is $\alpha'\in [0,1]$ such that $\alpha = \alpha'(1-Q[B]/n)$. We find $\alpha'\notin F^n$ because $\alpha\notin\bar F$. Thus, we have 
\begin{align*}
  \mathcal C_0&\ni \alpha'Y_0^n+(1-\alpha')X_0^*\\
  & = X_0^*+\frac{\alpha}{1-Q[B]/n}(Y_0^n-X_0^*) \\
  &= X_0^* + \frac{\alpha}{1-Q[B]/n}(X_0-X_0^*)(1-Q[B]/n) \\
  &= X_0^* + \alpha(X_0-X_0^*) \\
  &= X_0^\alpha.   
\end{align*} 
 Consequently, $\alpha\notin F_0$, and we have shown that~$F_0$ is countable. 
\end{proof}

\section{Integer superhedging}\label{a:integer superhedging}

In this section we discuss some properties of the integer superhedging price
$\sigma_{\mathbb Z}(C)$ of a claim, as defined in~\eqref{eq:def super}.
First, we give a simple example where it does not agree with the classical
superhedging price $\sup \Pi(C)$.

\begin{example}\label{ex:gap}
  In this example, the gap between $\sup\Pi(C)$ and the cheapest
  integer superhedging price $\sigma_\mathbb{Z}(C)$ has size~$a$, for an arbitrary number $a>0$.
  On the probability space $\Omega=\{\omega_1,\omega_2\}$,
  consider the one-dimensional model
  \[
    S^1_0 = 1, \quad S^1_1(\omega_1) = 1-2a, \quad  S^1_1(\omega_2)=1+2a
  \]
  with $r=0$. The unique equivalent martingale measure is
  $(\delta_{\omega_1}+\delta_{\omega_2})/2$, and so the unique arbitrage free price
  of the claim
  \[
    C(\omega_1)=0, \quad C(\omega_2)=2a
  \]
  is given by $\Pi(C)=\{a\}$. By Theorem~\ref{thm:int inclusion},
  we have $\Pi_{\mathbb{Z}}(C)=\Pi(C)=\{a\}$. The  integer superhedging
  price is found by computing
  \begin{align}
    \sigma_\mathbb{Z}(C)&=
      \inf_{\phi\in\mathbb{Z}}\max_{\omega\in\Omega}\big(C(\omega)
        -\phi \Delta S_1^1(\omega)\big)  \label{eq:min C} \\
    &=\min_{\phi\in\mathbb{Z}}\max\{2a\phi,2a-2a\phi\}=2a. \label{eq:min C2}
  \end{align}
  We obtain that the interval of prices of integer superhedges
  is $[2a,\infty)$.
  For \emph{real}~$\phi$, the minimum in~\eqref{eq:min C2}
  is attained at $\phi=\tfrac12$, yielding the classical superhedging
  price $a=\sup \Pi(C)$.
\end{example}

%Note that, for a $d$-dimensional one-period model, the cheapest integer superhedging price is given by
%$\inf_{\phi\in\mathbb{Z}^d}\max_{\omega\in\Omega}(C(\omega)/(1+r)-\phi \Delta \hat{S}_1(\omega))$.
%According to Proposition~\ref{prop:sh 1d}, the inf is attained for $d=1$,
%and so we know a priori that we may write $\min_{\phi\in\mathbb{Z}}$ 
%instead of $\inf_{\phi\in\mathbb{Z}}$ in~\eqref{eq:min C}.

As soon as a model is fixed, the gap considered in the preceding example
can be bounded for \emph{all} claims. In Example~\ref{ex:gap},
we have equality in~\eqref{eq:bound}. On~$\mathbb{R}^n$, we always use the
Euclidean norm $\|\cdot\|=\|\cdot\|_2$.

\begin{proposition}\label{prop:bound}
  Assume $\bf (F)$ and $\bf{NA}$, and let~$C$ be a claim. Then
  \begin{equation}\label{eq:bound}
    \sigma_{\mathbb Z}(C) - \sup \Pi(C) \leq
    \frac12 \sqrt{d}\, \max_{\omega\in\Omega} \sum_{k=1}^T\|\Delta \hat{S}_k (\omega)\|.
  \end{equation}
\end{proposition}
\begin{proof}
  Let $\bar\psi\in\mathcal R$ be a cheapest classical superhedge. By Theorem~\ref{thm:super},
  it satisfies $V_0(\bar\psi)=\sup \Pi(C)$.
  By rounding the risky positions of~$\bar\psi$
  to the closest integers (with any convention for half-integers),
  we get a strategy $(\psi^0,\lfloor\psi\rceil)\in\mathcal{Z}$. Clearly,
  \[
    \|\psi_t-\lfloor \psi \rceil_t\|\leq \|(\tfrac12,\dots,\tfrac12)\|=\tfrac12 \sqrt d,
    \quad t\in\mathbb{T}.
  \]
  Define $\hat C=C/(1+r)^T$, and let $\bar\phi\in\mathcal Z$.
  Since
  \[
    \hat{V}_T(\bar\phi) = V_0(\bar \phi) + \sum_{k=1}^T \phi_k \Delta \hat{S}_k,
  \]
  we get the necessary condition
  \begin{equation}\label{eq:V0 sh}
    V_0(\bar \phi)=\max_{\omega\in\Omega}\Big(
    \hat C(\omega)- \sum_{k=1}^T \phi_k(\omega) \Delta \hat{S}_k(\omega) \Big),
  \end{equation}
  if $\bar \phi$ should be a cheapest integer superhedge. It follows that
  \begin{align*}
    \sigma_{\mathbb Z}(C) &= \inf_{\substack{\phi\ \text{predictable,}\\ \mathbb{Z}^d\text{-valued}}}
      \ \max_{\omega\in\Omega}\Big(
      \hat C(\omega)- \sum_{k=1}^T \phi_k(\omega) \Delta \hat{S}_k(\omega) \Big) \\
    &\leq  \max_{\omega\in\Omega}\Big(
      \hat C(\omega)- \sum_{k=1}^T \lfloor\psi\rceil_k(\omega) \Delta \hat{S}_k(\omega) \Big) \\
    &\leq \max_{\omega\in\Omega}\Big(
      \hat C(\omega)- \sum_{k=1}^T \psi_k(\omega) \Delta \hat{S}_k(\omega) \Big)
        + \max_{\omega\in\Omega} \sum_{k=1}^T\big(\psi_k(\omega)-\lfloor\psi\rceil_k(\omega)\big)
           \Delta \hat{S}_k(\omega) \\
    &\leq \sup \Pi(C) + \max_{\omega\in\Omega} \sum_{k=1}^T
       \|\psi_k(\omega)-\lfloor\psi\rceil_k(\omega)\|\cdot \|\Delta \hat{S}_k(\omega)\| \\
    &\leq \sup \Pi(C) +\frac12 \sqrt{d}\, \max_{\omega\in\Omega} \sum_{k=1}^T\|\Delta \hat{S}_k (\omega)\|.
  \end{align*}
\end{proof}
The following example shows that, contrary to the case of classical superhedging,
there need not exist a cheapest integer superhedge.
\begin{example}\label{ex:no cheapest}
  Let $\Omega=\{\omega_1,\omega_2\}$, $r=0$, $T=1$ and $d=2$. We choose the model with $S_0 :=(2,2)$ and
   $$ S_1(\omega_j) = \begin{cases} (3,2-\sqrt{2}) &j=1,\\ (1,2+\sqrt{2}) &j=2.\end{cases} $$
  This model satisfies {\bf NA}, and it is complete in the classical sense. Indeed, the only martingale measure is given by $Q[\{\omega_j\}] = 1/2$ for $j=1,2$.
  Consider the claim
  \[
    C(\omega_1) = 1-\tfrac12\sqrt{2}, \quad C(\omega_2) =1+ \tfrac12\sqrt{2},
  \]
  whose set of (integer) arbitrage free prices is the singleton $\Pi_\mathbb{Z}(C)=\Pi(C)=\{1\}$.
  Then there is no minimizer for the superhedging problem (see~\eqref{eq:V0 sh})
  \[
     \inf_{\phi\in\mathbb{Z}^2} \max_{i=1,2}
       (C(\omega_i)-\phi \Delta S_1(\omega_i)) =:  \inf_{\phi\in\mathbb{Z}^2} f(\phi).
  \]
  Indeed, for $\phi\in\mathbb{R}^2$, the set of minimizers would be
  \[
    \phi \in \{(x,\tfrac12 + x/\sqrt{2}): x\in\mathbb{R}\},
  \]
  yielding $\inf_{\phi\in\mathbb{R}^2} f(\phi)=1$.
  Obviously, this set contains no integer strategies. By Kronecker's approximation
  theorem (Theorem~\ref{thm:kron}), the sequence $(\tfrac12+m/\sqrt{2})\bmod 1$, $m\in\mathbb{N}$,
  is dense in $[0,1]$. Thus, there is a sequence $m_k\in\mathbb{N}$
  such that
  \[
     0\leq (\tfrac12+m_k/\sqrt{2})\bmod 1 \leq \frac1k, \quad k\in\mathbb{N}.
  \]
  Define
  \[
    \phi^{(k)}:=\big(m_k,\lfloor\tfrac12+m_k/\sqrt{2}\rfloor\big) \in \mathbb{Z}^2, \quad k\in\mathbb{N}.
  \]
  Since $f$ is Lipschitz continuous (with constant~$L$, say), we have
  \begin{align*}
    |f(\phi^{(k)})-1| &= |f(\phi^{(k)})- f(m_k,\tfrac12+m_k/\sqrt{2})| \\
    &\leq L \big\|\big(0,(\tfrac12+m_k/\sqrt{2})\bmod 1\big) \big\| \to 0, \quad k\to\infty.
  \end{align*}
  Thus, the infimum of the prices of integer superhedges is $\sigma_{\mathbb Z}(C)=1$, but there
  is no cheapest integer superhedge.
\end{example}

Financial institutions usually hedge large portfolios of identical (or at least similar)
options. The following theorem shows that, when superhedging~$N$ copies of~$C$,
the integer superhedging price
per claim converges to the classical superhedging price: $\lim_{N\to\infty}N^{-1}\sigma_{\mathbb{Z}}(NC) = \sup \Pi(C).$ The second part of Theorem~\ref{thm:many claims} gives an estimate
on superhedging~$C$ with \emph{rational} strategies with controlled denominators.
\begin{theorem}\label{thm:many claims}
  Assume {\bf (F)} and {\bf NA}, and let~$C$ be a claim. Then
  \begin{itemize}
  \item[(i)]
  \[
    \frac{\sigma_{\mathbb{Z}}(NC)}{N} = \sup \Pi(C) + O\Big( \frac{1}{N}\Big),
    \quad N\to\infty.
  \]
  \item[(ii)]
  There is a sequence of rational strategies $\bar \psi^{(N)}\in\mathcal Q$
  such that all denominators occurring in~$\bar \psi^{(N)}$ have absolute value at most~$N$,
  \[
    V_0(\bar \psi^{(N)}) = \sup \Pi(C) + O\big(N^{-1/(nd(T+1))}\log N\big),
  \]
  and $\bar \psi^{(N)}$ is a superhedging strategy for~$C$.
  \end{itemize}
\end{theorem}
\begin{proof}
  (i) Assumption {\bf (F)} implies that the classical superhedging price
  $\sup \Pi(C)=\sup\{E_Q[C/(1+r)^T]:Q\in\mathfrak{Q}\}$ is finite.
  It is clear that
  \[
    N^{-1}\sigma_{\mathbb{Z}}(NC) \geq N^{-1} \sup \Pi(NC) = \sup \Pi(C)
  \]
  for all $N$. For the converse estimate, let $\bar\phi$ be a classical superhedging
  strategy for~$C$ with price $\sup \Pi(C)$ (see Theorem~\ref{thm:super}).
  Define
  \[
    \eta^{(N),j}_t(\omega) := \lfloor N \phi_t^j(\omega) \rfloor = N \phi_t^j(\omega)+O(1),
    \quad \omega\in\Omega,t\in\mathbb{T},1\leq j\leq d,N\in\mathbb N.
  \]
  We choose an arbitrary map $f:\mathbb N\to \mathbb R$ satisfying
  $\lim_{N\to\infty}f(N)=\infty$ and put
  \[
    \eta_1^{(N),0} := N\phi_1^0+f(N), \quad N\in\mathbb{N}.
  \]
  Then we  define   $\eta_t^{(N),0}$ for $t=2,\dots,T$ recursively to obtain a self-financing
  integer strategy $\bar\eta^{(N)}$ for each~$N$. By the definition of $\eta^{(N)}$,
  and since $\bar\phi$ is a superhedging strategy, we have
  \begin{align*}
    \frac{\hat{V}_T(\bar{\eta}^{(N)})}{N} &= \frac{V_0(\bar{\eta}^{(N)})}{N}
        +\sum_{k=1}^{T-1}\frac{\eta^{(N)}_k}{N} \Delta\hat{S}_k\\
     &=   \phi_1^0+ \frac{f(N)}{N} + \frac{\eta_1^{(N)} S_0}{N}
        +\sum_{k=1}^{T-1}\phi_k \Delta\hat{S}_k + O\Big(\frac1N\Big) \\
      &= \frac{f(N)}{N} +V_0(\bar\phi)
        +\sum_{k=1}^{T-1}\phi_k \Delta\hat{S}_k + O\Big(\frac1N\Big) \\
      &\geq \frac{C}{(1+r)^T} + \frac{f(N)}{N} +O\Big(\frac1N\Big)
      \geq \frac{C}{(1+r)^T} 
  \end{align*}
  for large~$N$. This shows that $\bar\eta^{(N)}$ is an integer superhedging strategy
  of $NC$ for large~$N$, and hence
  \begin{align*}
    \sigma_{\mathbb{Z}}(NC) &\leq V_0(\bar \eta^{(N)}) \\
     &= N V_0(\bar \phi) + O(f(N)) \\
      &= N \sup \Pi(C) + O(f(N)),\quad
     N\to\infty.
  \end{align*}
  It is easy to see that a quantity that is $O(f(N))$ for any~$f$ tending to infinity
  is~$O(1)$. Since~$f$ was arbitrary, the statement follows.
  
  (ii)   Again, let $\bar\phi$ be a classical superhedging
  strategy for~$C$ with price $\sup \Pi(C)$.
  By Dirichlet's approximation theorem (Theorem~\ref{thm:dir}), there are $1\leq q(N)\leq N$
  and $p(N,t,j,l)\in\mathbb Z$ such that
  \[
    |\phi_t^j(\omega_l)q(N)-p(N,t,j,l)| < N^{-1/(nd(1+T))}, \quad
    1\leq l\leq n,t\in\mathbb{T},1\leq j\leq d,N\in\mathbb N.
  \]
  We define
  \[
    \psi_t^{(N),j}(\omega_l) := \frac{p(N,t,j,l)}{q(N)}, \quad
    1\leq l\leq n,t\in\mathbb{T},1\leq j\leq d,N\in\mathbb N,
  \]
  which yields
  \[
    |\phi_t^j(\omega_l)-\psi_t^{(N),j}(\omega_l)|<N^{-1/(nd(1+T))}, \quad
    1\leq l\leq n,t\in\mathbb{T},1\leq j\leq d,N\in\mathbb N.
  \]
  After fixing the initial bank account position
  \[
    \psi_1^{(N),0}:= \phi_1^0 + N^{-1/(nd(1+T))}\log N, \quad N\in\mathbb N,
  \]
  a strategy $\bar \psi^{(N)}\in\mathcal Q$ is defined for each~$N$.
  By definition,
  \begin{align*}
    V_0(\bar \psi^{(N)}) &= \phi_1^0 + N^{-1/(nd(1+T))}\log N + \psi_1^{(N)}S_0 \\
    &= \phi_1^0 + N^{-1/(nd(1+T))}\log N + \phi_1 S_0 + O\big(N^{-1/(nd(1+T))}\big) \\
    &= \sup\Pi(C) + O\big(N^{-1/(nd(1+T))}\log N\big), \quad N\to\infty.
  \end{align*}
  It remains to show that $\bar \psi^{(N)}$ is a superhedge for~$C$ for large~$N$.
  This follows from
  \begin{align*}
    \hat{V}_T(\bar \psi^{(N)}) &= V_0(\bar \psi^{(N)}) + \sum_{k=1}^T \psi_k^{(N)} \Delta\hat{S}_k \\
    &= V_0(\bar \phi) + N^{-1/(nd(1+T))}\log N
      + \sum_{k=1}^T \phi_k \Delta\hat{S}_k + O\big(N^{-1/(nd(1+T))}\big) \\
    &\geq \frac{C}{(1+r)^T} + N^{-1/(nd(1+T))}\log N + O\big(N^{-1/(nd(1+T))}\big) \\
    &\geq   \frac{C}{(1+r)^T}, \quad N\ \text{large.}
  \end{align*}
 For those finitely many $N$ where the last inequality does not hold, we can simply add a sufficient amount of initial capital to obtain a superhedge; this does not change the convergence rate.
\end{proof}
From the proof of~(ii), it is clear that~$\log N$ can be replaced by an arbitrary
function tending to infinity.

\section{Variance optimal hedging in one period}\label{se:var opt}

We consider a one-period model satisfying {\bf (F)} and {\bf NA}.
Moreover, we suppose that $d\leq n$.
Our goal is to \emph{approximately} hedge a given (non-replicable) claim~$C$. For tractability, the error is measured
by the norm of $L^2(P^*)$, where~$P^*$ is a fixed EMM; we denote this norm by $\|\cdot \|$ throughout this section.
In the classical case, this leads to the optimization problem
\begin{equation}\label{eq:class var}
   \inf_{\phi\in\mathbb{R}^d} \inf_{V_0 \in\mathbb{R}}\| C/(1+r)-V_0-\phi \Delta S_1 \|
  = \inf_{\phi\in\mathbb{R}^d} \| \tilde{C} - \phi \Delta S_1 \|,
\end{equation}
where $\tilde{C}:=(C-E^*[C])/(1+r)$. Note that $\inf_{V_0 \in\mathbb{R}}$ is attained
at
\[
  V_0=E^*[C/(1+r)-\phi \Delta S_1]=E^*[C]/(1+r).
\]
The problem~\eqref{eq:class var} is then solved
by projecting~$\tilde{C}$ orthogonally to the space
$\{\phi \Delta S_1 : \phi\in\mathbb{R}^d\}$, which is closed
by Theorem~6.4.2 in~\cite{DeSc06}.
For more details on variance-optimal hedging (in particular, on the multi-period problem),
we refer to Chapter~10 of~\cite{FoSc16} and the references given there.

Now we proceed to our setup, and restrict~$\phi$ to~$\mathbb{Z}^d$. The minimization w.r.t.~$V_0$
is done as in~\eqref{eq:class var}, and we thus have to compute
\begin{equation}\label{eq:var opt}
  \inf_{\phi\in\mathbb{Z}^d} \| \tilde{C} - \phi \Delta S_1 \|.
\end{equation}
We have
\begin{align*}
  \| \tilde{C} - \phi \Delta S_1 \|^2 &= \sum_{l=1}^n P^*[\omega_l]\big(\tilde{C}(\omega_l)-\phi \Delta S_1(\omega_l)\big)^2\\
  &= \sum_{l=1}^n \big(\tilde{C}(\omega_l)P^*[\omega_l]^{1/2}-\phi \Delta S_1(\omega_l)P^*[\omega_l]^{1/2}\big)^2.
\end{align*}
The problem~\eqref{eq:var opt} thus amounts to computing the element of the lattice
\begin{equation}\label{eq:lattice}
  \left\{ \phi^1
  \left(
  \begin{matrix}
     \Delta S_1^1(\omega_1)P^*[\omega_1]^{1/2} \\
     \vdots \\
     \Delta S_1^1(\omega_n)P^*[\omega_n]^{1/2}
  \end{matrix}
  \right)
   +\dots + \phi^d
   \left(
  \begin{matrix}
     \Delta S_1^d(\omega_1)P^*[\omega_1]^{1/2} \\
     \vdots \\
     \Delta S_1^d(\omega_n)P^*[\omega_n]^{1/2}
  \end{matrix}
  \right)
   : \phi\in\mathbb{Z}^d \right\} \subset \mathbb{R}^n
\end{equation}
closest to the vector
\begin{equation}\label{eq:point}
  \left(
  \begin{matrix}
     \tilde{C}(\omega_1)P^*[\omega_1]^{1/2} \\
     \vdots \\
     \tilde{C}(\omega_n)P^*[\omega_n]^{1/2}
  \end{matrix}
  \right) \in \mathbb{R}^n
\end{equation}
w.r.t.\ the Euclidean norm.
This is an instance of the \emph{closest vector problem} (CVP), a well-known computational
problem with applications in cryptography, communications theory and other fields.
The survey paper~\cite{HaPuSt11} offers an accessible introduction to this subject
with many references.
By the Pythagorean theorem, the closest lattice point is the lattice 
point closest to the projection of~\eqref{eq:point} to the subspace generated by the lattice.
A cheap method to compute a (hopefully) close lattice point consists of rounding
the coefficients of this projected point to the closest integers. It is well-known, though,
that the resulting point may be far from optimal. This happens in the following example.

We consider
a toy example with $d=2$, $|\Omega|=4$, and $r=0$, specified in Table~\ref{ta:ex}.
The numbers are not calibrated to any market data, but are chosen to illustrate the point that
a naive approach at integer approximate hedging (as mentioned above) can lead to significant errors.
A detailed investigation of integer variance-optimal hedging over several periods in realistic
models is left to future work.

\begin{table}[h]
  \begin{center}
    \begin{tabular}{|c|c|c|c|c|}
    \hline
      $i$ & 1 & 2 & 3 & 4 \\ \hline
      $P^*[\omega_i]$ & $0.37$ & $0.18$ & $0.4$ & $0.05$ \\ \hline
      $\Delta S^1(\omega_i)$ & $-9$ & $-9$ & $0$ & $99$ \\ \hline
      $\Delta S^2(\omega_i)$ & $10$ & $1$ & $4$ & $-109.6$ \\ \hline
      $C(\omega_i)$ & $0$ & $7$ & $1$ & $8$ \\ \hline
    \end{tabular}
    \caption{Model parameters, risk neutral measure, and a claim.}\label{ta:ex}
  \end{center}
\end{table}

We wish to approximately hedge~$N$ copies of the claim, i.e., the claim~$NC$, for $N\in\mathbb N$.
First, we computed the classical variance-optimal hedge~$\phi^{(N)}=N \phi^{(1)}\in\mathbb{R}^2$ by
projection (see~\eqref{eq:class var}). The relative $L^2(P^*)$-error
\[
  \frac{\|N\tilde{C}- \phi^{(N)}\Delta S_1 \|}{\|N\tilde{C}\|}
    = \frac{\|\tilde{C}- \phi^{(1)}\Delta S_1 \|}{\|\tilde{C}\|},
\]
which of course does not depend on~$N$, is displayed in the second line of Table~\ref{ta:results}.
The third line of Table~\ref{ta:results} contains the maximal position size
$\max_{i=1,2}|\phi^{(N),i}|=N \max_{i=1,2}|\phi^{(1),i}|$
in the underlying assets.
Then, we solved the integer variance-optimal hedging
problem~\eqref{eq:var opt} exactly, using the algorithm CLOSESTPOINT described
in~\cite{AgErVaZe02}, which is based on the Schnorr-Euchner algorithm~\cite{ScEu94}.
CVP is known as a computationally hard problem, with the fastest algorithms having
exponential complexity in the dimension. Since our dimension is only $|\Omega|=4$, this
was not an issue in this toy example.
In more sophisticated examples, a preprocessing using the LLL-algorithm
\cite{NgVa10} might faciliate the task of computing a closest vector.
Table~\ref{ta:results} shows the relative $L^2(P^*)$-error
\[
  \frac{\|N\tilde{C}- \phi_{\mathrm{CVP}}^{(N)}\Delta S_1 \|}{\|N\tilde{C}\|}
\]
and the maximum position size.
Finally, we used a poor man's approach at
solving~\eqref{eq:var opt} approximately, by simpling rounding the positions
of the classical hedge $\phi^{(N)}$ to the closest integers. From Table~\ref{ta:results},
we see that this works fine for large~$N$, but gives significantly worse results
than solving CVP for small~$N$.
Note that, in this example, the position sizes of the integer hedge are much smaller than that of the classical hedge.
Finally, we mention that computing the so-called covering radius~\cite{HaPuSt11} of the
lattice~\eqref{eq:lattice} yields an upper bound for the hedging error for any claim.
\begin{table}[h]
  \begin{center}
    \begin{tabular}{|c|c|c|c|c|c|c|c|}
    \hline
      $N$ & 1 & 5 & 10 & 20 & 30 & 40 & 50\\ \hline
      classical: rMSE & $0.405$ & $0.405$ & $0.405$ & $0.405$ & $0.405$ & $0.405$ & $0.405$  \\ \hline
      classical: position size  & $1.688$ & $8.438$ & $16.876$ & $33.752$ & $50.627$ & $67.503$ & $84.379$ \\ \hline
      CVP: rMSE & $0.901$ & $0.431$ & $0.419$ & $0.416$ & $0.415$ & $0.415$ & $0.410$ \\ \hline
      CVP: position size & $1$ & $3$ & $5$ & $11$ & $16$ & $21$ & $25$ \\ \hline
      rounding: rMSE & $8.352$ & $1.636$ & $0.419$ & $0.419$ & $0.419$ & $0.419$ & $0.412$ \\ \hline
      rounding: position size & $1$ & $3$ & $5$ & $10$ & $15$ & $20$ & $26$\\ \hline
    \end{tabular}
    \caption{Errors and position sizes for variance optimal hedging.}\label{ta:results}
  \end{center}
\end{table}

\appendix

\section{Tools from number theory}\label{app:nt}

In this appendix we collect the classical number theoretic theorems we have used
in this paper. The theorems of Dirichlet and Kronecker are fundamental results
in Diophantine approximation (i.e., the approximation of real numbers by rational numbers).

\begin{theorem}[Dirichlet's approximation theorem; Theorem 1B, Chapter II in~\cite{sc91}]\label{thm:dir}
  Given $\alpha_1,\dots,\alpha_n\in\mathbb R$ and an integer $N>1$, there are
  $q,x_1,\dots,x_n\in\mathbb Z$ with $1\leq q\leq N$ and
  \[
    |\alpha_i q-x_i| < N^{-1/n}, \quad 1\leq i\leq n.
  \]
\end{theorem}

\begin{theorem}[Kronecker's approximation theorem; Theorem~7.7 in~\cite{Ap90}]\label{thm:kron}
  If~$\theta\in\mathbb R$ is an irrational number, then the sequence $(n \theta \bmod 1)_{n\in\mathbb N}$ is dense in~$[0,1]$.
\end{theorem}

We also mention here the following classical theorem~\cite{Gr87}:
\begin{theorem}[Minkowski]\label{thm:mink}
  Let~$\mathcal{K}\subset \mathbb{R}^d$ be closed, convex, zero-symmetric, and bounded.
  If the volume of $\mathcal{K}$ satisfies $\mathrm{vol}(\mathcal{K}) \geq 2^d,$
  then $\mathcal{K}$ contains a non-zero point with integral coordinates.
\end{theorem}
We did not apply Theorem~\ref{thm:mink} in the rest of the paper, but hint at
a possible application.
Consider the following one-period portfolio optimization problem with
maximum loss constraint, where $c>0$ and~$U$ is some utility function:
\[
   E[U(\phi \Delta S_1)] \to \max!\qquad \phi \in\mathbb{Z}^d,\ \phi \Delta S_1 \geq -c.
\]
Then, Theorem~\ref{thm:mink} easily yields a sufficient criterion to ensure that the admissibility
set contains a non-zero portfolio. Refinements of Theorem~\ref{thm:mink} give
several linearly independent portfolios.

%\begin{lemma}
%  Let $q\in\mathbb Q\setminus\{0\}$. Then $e^q\notin\mathbb Q$. In particular, the function
%   $$ f:(0,1)\rightarrow\mathbb R^2, x\mapsto (x,e^{1-1/x}) $$
%  attains no values in $\mathbb Q^2$. 
%\end{lemma}
%\begin{proof}
%  
%\end{proof}

%\section{Questions, To-do list}

%{\bf Stuff we should do:}

%\begin{enumerate}
%   \item Explicitly state open problems
%\end{enumerate}

%{\bf Stuff we might do:}

%\begin{enumerate}
%  \item Give a counterexample for (iii) $\Rightarrow$ (ii) in Theorem~\ref{t:NA NIFLVR}
%  with unbounded~$Y$
%   \item draw some pictures
%   \item If the model values are rational, is there a cheapest integer superhedge? \margincomment{A cheapest superhedge has rational values if all model parameters are rational. There might be cheapest superhedges with non-rational values though.}
%\end{enumerate}

%\bibliographystyle{alpha}
\bibliographystyle{siam}
\bibliography{gerhold}

\end{document}